\documentclass[12pt]{article}
\usepackage{amsthm,amsmath}
\usepackage[margin=1in]{geometry}
\usepackage{natbib}
\usepackage{titlesec}
\usepackage{dsfont, amssymb, amscd, caption,  color}
\usepackage{graphicx,subfigure}
\usepackage{framed}
\usepackage[titletoc,title]{appendix}

\usepackage{times,soul}
\usepackage{bm}
\titleformat{\section}{\bfseries\centering}{\thesection.}{0.5em}{}

\usepackage{tikz}

\newtheorem{theorem}{Theorem}
\newtheorem{lemma}{Lemma}

\newtheorem{proposition}{Proposition}
\theoremstyle{definition}
\newtheorem{definition}{Definition}
\newtheorem{assumption}{Assumption}
\newtheorem{condition}{Condition}

\theoremstyle{remark}
\newtheorem{remark}{Remark}

\newcommand{\bpm}{\begin{pmatrix}}
\newcommand{\epm}{\end{pmatrix}}
\newcommand{\bbm}{\begin{bmatrix}}
\newcommand{\ebm}{\end{bmatrix}}
\def\T{{ \mathrm{\scriptscriptstyle T} }}
\DeclareMathOperator*{\subjectto}{subject\ to}
\newcommand{\vertiii}[1]{{\left\vert\kern-0.25ex\left\vert\kern-0.25ex\left\vert #1 
    \right\vert\kern-0.25ex\right\vert\kern-0.25ex\right\vert}}
\usepackage{setspace}

\begin{document}

\title{Network Reconstruction From High Dimensional Ordinary Differential Equations}
\author{Shizhe Chen, Ali Shojaie, and Daniela M. Witten\footnote{ 
Shizhe Chen is Graduate Student, Department of Biostatistics, University of Washington, WA 98195 (e-mail: shizhe.chen@gmail.com); 
Ali Shojaie is Associate Professor, Department of Biostatistics, and Adjunct Associate Professor, Department of Statistics, University of Washington, WA 98195 (e-mail: ashojaie@u.washington.edu); 
and Daniela M. Witten is Associate Professor,  Departments of Biostatistics and Statistics, University of Washington, WA 98195  (e-mail: dwitten@u.washington.edu). 
We thank the associate editor and two anonymous reviewers for helpful comments. 
We thank the authors of \cite{brunel2014}, \cite{hall2014}, \cite{henderson2014},  and \cite{wu2014} for sharing their code for their proposals, and for responding to our inquiries. 
We thank the Allen Institute for Brain Science for providing the data set analyzed in Section~\ref{sec::CI}.
A.S. was supported by NSF grant DMS-1561814 and NIH grants 1K01HL124050-01A1 and 1R01GM114029-01A1, and  D.W. was supported by NIH Grant DP5OD009145, NSF CAREER Award DMS-1252624, and an Alfred P. Sloan Foundation Research Fellowship. 
} 
}
\date{}
\maketitle

\begin{center}
Department of Biostatistics \\
University of Washington\\
Box 357232\\
Seattle, WA 98195-7232\\
\end{center}

\begin{abstract}
We consider the task of learning a dynamical system from high-dimensional time-course data. For instance, we might wish to estimate a gene regulatory network from gene expression data measured at discrete time points. We model the dynamical system non-parametrically as a system of additive ordinary differential equations. Most existing methods for parameter estimation in ordinary differential equations estimate the derivatives from noisy observations. This is known to be challenging and inefficient. We propose a novel approach that does not involve derivative estimation. We show that the proposed method can consistently recover the true network structure even in high dimensions, and we demonstrate empirical improvement over competing approaches.  


\end{abstract}

\section{INTRODUCTION}\label{sec::intro}

Ordinary differential equations (ODEs) have been widely used to model dynamical systems in many fields, including chemical engineering \citep{biegler1986},  genomics \citep{chou2009}, neuroscience \citep{izhikevich2007}, and infectious diseases \citep{wu2005}.  A system of ODEs takes the form
\begin{equation}
X'(t;\theta)\equiv \bbm \frac{d X_1(t;\theta)}{dt} \\ \vdots \\ \frac{dX_p(t;\theta)}{dt} \ebm = \bbm f_1(X(t;\theta), \theta ) \\ \vdots \\  f_p(X(t;\theta), \theta ) \ebm  \equiv f(X(t;\theta),\theta); \quad t \in [0,1],
\label{eqn::ODE_model}
\end{equation} 
where $X(t;\theta)=(X_1(t;\theta), \ldots, X_p(t;\theta))^{\T}$ denotes a set of variables, and the form of the functions $f=(f_1, \ldots, f_p)^{\T}$ may be known or unknown. 
In \eqref{eqn::ODE_model}, $t$ indexes time. 
Typically, there is also an initial condition of the form $X(0;\theta)=C$, where $C$ is a $p$-vector.  
In practice, the system \eqref{eqn::ODE_model} is often observed on discrete time points subject to measurement errors. 
Let $Y_i \in \mathbb{R}^p$ be the measurement of the system at time $t_i$ such that 
\begin{equation} \label{eqn::noisy}
Y_i= X(t_i;\theta^*)+\epsilon_i, \quad i=1, \ldots, n,
\end{equation}
where $\theta^*$ denotes the true set of parameter values and the random $p$-vector $\epsilon_i$ represents independent measurement errors. 
In what follows, for notational simplicity, we sometimes suppress the dependence of $X(t;\theta)$ on $\theta$, i.e., $X(t) \equiv X(t;\theta)$ in \eqref{eqn::ODE_model} and  $X^*(t) \equiv X(t;\theta^*)$ in \eqref{eqn::noisy}. 

In the context of high-dimensional time-course data arising from biology, it can be of interest to recover the structure of a system of ODEs --- that is, to determine which features regulate each other.
If $f_j$ in \eqref{eqn::ODE_model} is a function of $X_k$, then we say that $X_k$ \emph{regulates} $X_j$ in the sense that $X_k$ controls the changes of $X_j$ through its derivative $X'_j$.
For instance, biologists might want to infer gene regulatory networks from noisy time-course gene expression data.
In this case, the number of variables $p$ exceeds the number of time points $n$; we refer to this as the high-dimensional setting. 

In high-dimensional statistics, sparsity-inducing penalties such as the lasso \citep{tibshirani1996} and the group lasso \citep{yuan2006} have been well-studied. Such penalties have also been extensively used to recover the structure of probabilistic graphical models (e.g., \citealp{yuan2007,friedman2008, meinshausen2010, voorman2014}). 
However, model selection in high-dimensional ODEs remains a relatively open problem, with the exception of some notable recent work \citep{lu2011,henderson2014,wu2014}.  
In fact, the tasks of parameter estimation and model selection in ODEs from noisy data are very challenging, even in the classical statistical setting where $n > p$ (see e.g., 
\citealp{ramsay2007, brunel2008, liang2008,qi2010,  xue2010, gugushvili2012, hall2014, zhang2015}). Moreover, the problem of high-dimensionality is compounded if the form of the function $f$ in \eqref{eqn::ODE_model} is unknown, leading to both statistical and computational issues. 

In this paper, we propose an efficient procedure for structure recovery of an ODE system of the form \eqref{eqn::ODE_model} from noisy observations of the form \eqref{eqn::noisy}, in the setting where the functional form of $f$ is unknown. In Section~\ref{sec::method}, we review existing methods. In Section~\ref{sec::ourmethod}, we propose a new structure recovery procedure.  In Section~\ref{sec::theory}, we study the theoretical properties of our proposal. In Section~\ref{sec::simulation}, we apply our procedure to simulated data. In Section~\ref{sec::RDA}, we apply it to \textit{in silico} gene expression data generated by GeneNetWeaver \citep{schaffter2011} and to calcium imaging data.  We conclude with a discussion in Section~\ref{sec::discussion}. Proofs and additional details are provided in the supplementary material.

\section{LITERATURE REVIEW}\label{sec::method}
In this section, we review existing statistical methods for parameter estimation and/or model selection in ODEs. 
Most of the methods reviewed in this section are proposed for the low-dimensional setting.
Even though they may not be directly applicable to the high-dimensional setting, they lay the foundation for the development of model selection procedures in high-dimensional additive ODEs.

\subsection{Notation}

Without loss of generality, assume that  $0=t_1 < t_2 < \ldots < t_n=1$.  We let $Y_{ij}$ indicate the observation of the $j$th variable at the $i$th time point, $t_i$. We use $\mathcal{X}(h)$ to denote a nonparametric class of functions on $[0,1]$ indexed by some smoothing parameter(s) $h$. We use $Z(\cdot)$ to represent an arbitrary function belonging to $\mathcal{X}(\cdot)$. We use $\|\cdot\|_2$ to denote the $\ell_2$-norm of a vector or a matrix, and $\vertiii{f}$ to denote the $\ell_2$-norm of a function $f$ on the interval $[0,1]$, i.e. $\vertiii{f}^2 \equiv \int_0^1 f^2(t) \, dt$.  
We use an asterisk to denote true values---for instance, $\theta^*$ denotes the true value of $\theta$ in \eqref{eqn::ODE_model}.
We use $\Lambda_{\min}(A)$ and $\Lambda_{\max}(A)$ to denote the minimum and maximum eigenvalues of a square matrix $A$, respectively.

\subsection{Methods that assume a known form of $f$}\label{sec::knownf}

\subsubsection{Gold standard approach}\label{sec::gold}

To begin, we suppose that the function $f$ in \eqref{eqn::ODE_model} takes a known form. \cite{benson1979} and \cite{biegler1986} proposed to estimate the unknown parameter $\theta^*$ in \eqref{eqn::noisy} by solving the problem 
\begin{subequations}
\label{eqn::gold}
\begin{align}
\hat{\theta}^{\text{gold}}= \underset{\theta }{\arg \min} & \ \sum_{i=1}^n \|Y_i - X(t_i;\theta)\|^2_2 \label{eqn::gold_objective}\\      
\subjectto & \quad X'(t;\theta)= f(X(t;\theta), \theta ), \quad t \in [0,1]. \label{eqn::gold_constraint}
\end{align}
\end{subequations} 
Note that $X(\cdot;\theta)$ in \eqref{eqn::gold} is a fixed function given $\theta$, although an analytic expression may not be available.  
The resulting estimator $\hat{\theta}^{\text{gold}}$ has appealing theoretical properties: for instance, when the measurement errors $\epsilon_i$ in \eqref{eqn::noisy} are Gaussian, then $\hat{\theta}^{\text{gold}}$ is the maximum likelihood estimator, and is $\sqrt{n}$-consistent.
In this sense, \eqref{eqn::gold} can thus be considered  the \emph{gold standard} approach. 
However, solving \eqref{eqn::gold} is often computationally challenging.

\subsubsection{Two-step collocation methods}\label{sec::two-step}

In order to overcome the computational challenges associated with solving \eqref{eqn::gold}, \emph{collocation} methods have been employed by a number of authors \citep{varah1982, ellner2002,   ramsay2007, brunel2008, cao2008,liang2008, cao2011, lu2011, gugushvili2012,  brunel2014,  hall2014, henderson2014,    wu2014, dattner2015, zhang2015}. 

The two-step collocation procedure first proposed by \cite{varah1982} involves fitting a smoothing estimate $\hat{X}(\cdot;h)$ to the observations $Y_1,\ldots, Y_n$ in \eqref{eqn::noisy} with a smoothing parameter  $h$, and then plugging $\hat{X}(\cdot;h)$ and its derivative with respect to $t$ into \eqref{eqn::ODE_model} in order to estimate $\theta$. This amounts to solving the optimization problem 
\begin{subequations}
\label{eqn::two}
\begin{align}
\hat{\theta}^{\text{TS}} & =  \underset{\theta}{\arg \min} \ \int_0^1 \left\| \hat{X}'(t;{h})-f\big(\hat{X}(t;{h}),\theta\big)\right\|^2_2 \, dt, \label{eqn::two_objective}\\
\intertext{where }
& \hat{X}(\cdot;h)= \underset{ Z(\cdot) \in \mathcal{X}(h) }{\arg \min}  \sum_{i=1}^n \|Y_i - Z(t_i)\|_2^2. \label{eqn::two_constraint}
\end{align}
\end{subequations}
The two-step procedure \eqref{eqn::two} has a clear advantage over the gold standard approach \eqref{eqn::gold} because the former decouples the estimation of $\theta$ and $X$. 
However, this advantage comes at a cost: due to the presence of $\hat{X}'$ in \eqref{eqn::two_objective}, the properties of the estimator $\hat{\theta}^{\text{TS}} $ in \eqref{eqn::two} rely heavily on the smoothing estimates obtained in \eqref{eqn::two_constraint}, and $\sqrt{n}$-consistency has only been shown for certain values of the smoothing parameter $h$ that are hard to choose in practice \citep{brunel2008,liang2008,gugushvili2012}. 

\cite{dattner2015} proposed an improvement to \eqref{eqn::two} for a special case of \eqref{eqn::ODE_model}.
To be more specific,  they assume that $f_j(X(t),\theta)$ in \eqref{eqn::ODE_model} is a linear function of $\theta$, which leads to 
\begin{equation}
X'(t)\equiv \bbm \frac{d X_1(t)}{dt} \\ \vdots \\ \frac{dX_p(t)}{dt} \ebm = \bbm g_1^{\T}(X(t))\theta \\ \vdots \\  g_p^{\T}(X(t)) \theta\ebm   \equiv  g(X(t))\theta; \quad t \in [0,1],
\label{eqn::ODE_linear}
\end{equation}
where $g(X(t))$ is a known function of $X(t)$. 
Integrating both sides of \eqref{eqn::ODE_linear} gives 
\begin{equation}
X(t)= \left\{ \int_0^t g(X(u)) \, du \right\} \theta + C,
\label{eqn::ODE_linear_int}
\end{equation}
where $C\equiv X(0;\theta)$. The unknown parameter $\theta^*$ is estimated by solving 
\begin{subequations}
\label{eqn::linear}
\begin{align}
\hat{\theta}^{\text{LM}}= \underset{\theta}{\arg \min} & \int_0^1 \left\|\hat{X}(t;{h}) - \left\{\int_0^t g\big(\hat{X}(u;{h})\big) \, du\right\} \theta - C \right\|_2^2 \, dt, \label{eqn::linear_object}\\
\intertext{where }
& \hat{X}(\cdot;h)= \underset{Z(\cdot)\in \mathcal{X}(h)}{\arg \min} \sum_{i=1}^n \|Y_i- Z(t_i)\|_2^2.\label{eqn::linear_constraint}
\end{align}
\end{subequations}
The optimization problem \eqref{eqn::linear_object} has an analytical solution, given the smoothing estimates from \eqref{eqn::linear_constraint}. Compared with the two-step procedure \eqref{eqn::two}, this approach requires an estimate of the integral, $\int_0^t g\big(\hat{X}(u;{h})\big) \, du$ in \eqref{eqn::linear_object}, rather than an estimate of  the derivative, $\hat{X}'(t;h)$. This has profound effects on the asymptotic behaviour of the estimator $\hat{\theta}^{\text{LM}}$. $\sqrt{n}$-consistency of $\hat{\theta}^{\text{LM}}$ has been established under mild conditions, and it has been found that the choice of smoothing parameter $h$ is less crucial than for other methods \citep{gugushvili2012}.

Recently, \cite{brunel2014} and \cite{hall2014} have considered alternatives to the loss function in \eqref{eqn::two_objective}.
Let $\mathbb{C}^1 (0,1)$ be the set of functions that are first-order differentiable on $(0,1)$ and equal zero on the boundary points $0$ and $1$. 
Then \eqref{eqn::ODE_model} implies that, 
for any $\phi \in \mathbb{C}^1(0,1)$,
\begin{equation}\label{eqn::ODE_variation}
\int_0^1 f(X(t), \theta) \phi(t)dt + \int_0^1 X(t) \phi'(t)dt = 0.
\end{equation}
Equation \eqref{eqn::ODE_variation} is referred to as the \emph{variational formulation} of the ODE. 
A least squares loss based on \eqref{eqn::ODE_variation} takes the form 
\begin{equation}\label{eqn::var_objective}
\hat{\theta}^{\text{V}} =  \underset{\theta}{\arg \min}  \frac{1}{L}\sum_{l=1}^L \left\|\int_0^1 f\big( \hat{X}(t; h), \theta\big) \phi_l(t)dt + \int_0^1 \hat{X}(t; h) \phi'_l(t)dt \right\|^2_2, 
\end{equation}
where $\hat{X}(t;h)$ is defined in \eqref{eqn::two_constraint} and $\{ \phi_l, l=1,\ldots, L \}$ is a finite set of functions in $\mathbb{C}^1(0,1)$ \citep{brunel2014}.
In \cite{hall2014}, the loss function is the sum of the loss functions in \eqref{eqn::two_constraint} and \eqref{eqn::var_objective}, so that $\theta$ and the optimal bandwidth $h$ are estimated simultaneously. 
It is immediately clear that the derivative ${X}'(\cdot;\theta)$ is not needed in  \eqref{eqn::var_objective}, which can lead to substantial improvement compared to the two-step procedure in \eqref{eqn::two}. 
A minor drawback of \eqref{eqn::var_objective} is that the variational formulation \eqref{eqn::ODE_variation} is enforced on a finite set of functions $\{\phi_l, l=1,\ldots, L \}$ rather than on the whole class $\mathbb{C}^1 (0,1)$.
Under suitable assumptions, the estimator $\hat{\theta}^{\text{V}}$ is $\sqrt{n}$-consistent \citep{brunel2014,hall2014}.

\subsubsection{The generalized profiling method}\label{sec::profiling}

Another collocation-based method is the generalized profiling method of \cite{ramsay2007}. Instead of the smoothing estimate $\hat{X}(\cdot;h)$ in \eqref{eqn::two_constraint}, the generalized profiling method uses a smoothing estimate $\check{X}(\cdot;h,\theta)$ that minimizes the weighted sum of a data-fitting loss and a model-fitting loss for any given $\theta$. In greater detail, 
\begin{subequations}
\label{eqn::GPM}
\begin{equation}
 \hat{\theta}^{\text{GP}}_{\lambda}  = \underset{\theta}{\arg \min}  \sum_{i=1}^n \left\|Y_i-\check{X}(t_i;{h},\theta)\right\|_2^2, \label{eqn::GPM_objective}
\end{equation}
where 
\begin{equation}
 \check{X}(\cdot;h,\theta)= \underset{Z(\cdot) \in \mathcal{X}(h) }{\arg \min}  \frac{1}{n}\sum_{i=1}^n \|Y_i-Z(t_i)\|_2^2 + \lambda \int_0^1 \| {Z}'(t)- f({Z}(t),\theta) \|^2_2 \, dt. \label{eqn::GPM_constraint}
\end{equation}
\end{subequations}
In \cite{ramsay2007}, the authors solve \eqref{eqn::GPM_objective} iteratively for a non-decreasing sequence of $\lambda$'s in \eqref{eqn::GPM_constraint}. 
$\sqrt{n}$-consistency of the limiting estimator was later established by \cite{qi2010}.
\cite{zhang2015} proposed a model selection procedure by applying an \emph{ad hoc} lasso procedure \citep{wang2007} to the estimates from \eqref{eqn::GPM}.

\subsection{Methods that do not assume the form of $f$}\label{sec::highd}

A few authors have recently considered modeling large-scale dynamical systems from biology using ODEs \citep{henderson2014, wu2014}, under the assumption that the right-hand side of \eqref{eqn::ODE_model} is additive, 
\begin{equation}\label{eqn::additivity}
X'_j(t) = \theta_{j0}  + \sum_{k=1}^p f_{jk}(X_k(t)),  \quad \theta_{j0} \in \mathbb{R}.
\end{equation}
\cite{henderson2014} and \cite{wu2014} approximate the unknown $f_{jk}$ with a truncated basis expansion. 
Consider a finite basis,  $\psi(x)=( \psi_1(x), \ldots, \psi_{M}(x))^{\T}$, such that 
\begin{equation}\label{eqn::fjk_approx}
f_{jk}(a_k)= \psi(a_k)^{\T}\theta_{jk} + \delta_{jk}(a_k), \quad \theta_{jk} \in \mathbb{R}^M,
\end{equation}
where $\delta_{jk}(a_k)$ denotes the residual. 
Using \eqref{eqn::fjk_approx}, a system of additive ODEs of the form \eqref{eqn::additivity} can be written as 
\begin{equation}
X'_j(t) = \theta_{j0}+\sum_{k=1}^p  \psi(X_k(t))^{\T} \theta_{jk}+\sum_{k=1}^p \delta_{jk}(X_k(t)), \quad j=1,\ldots,p.
\label{eqn::ODE_basis_finite}
\end{equation} 
\cite{henderson2014} and \cite{wu2014} consider the problem of estimating and selecting the non-zero elements $\theta_{jk}$ in \eqref{eqn::ODE_basis_finite}. 
Roughly speaking, they propose to solve optimization problems of the form   
\begin{subequations}\label{eqn::np}
\begin{equation}
\begin{split}
\hat{\theta}^{\text{NP}}_j =  \underset{\theta_{j0} \in \mathbb{R}, \theta_{jk} \in \mathbb{R}^M }{\arg \min}  \  & \int_0^1 \left\| \hat{X}'_j(t;{h})-\theta_{j0}-\sum_{k=1}^p  \psi\big(\hat{X}_k(t;h)\big)^{\T} \theta_{jk}\right\|^2_2 \, dt \\
&+ \lambda_n \sum_{k=1}^p \left[ \int_0^1 \{ \psi\big(\hat{X}_k(t;h)\big)^{\T}\theta_{jk}\}^2 \, dt\right]^{1/2}, 
\end{split}
\label{eqn::np_objective}
\end{equation}
for $j=1,\ldots,p$, where
\begin{equation} 
  \hat{X}(\cdot;h)= \underset{ Z(\cdot) \in \mathcal{X}(h) }{\arg \min}  \sum_{i=1}^n \|Y_i - Z(t_i)\|_2^2. \label{eqn::np_constraint}
 \end{equation}
\end{subequations}
In \eqref{eqn::np_objective}, a standardized group lasso penalty  
forces all elements in $\theta_{jk}$ to be either zero or non-zero when $\lambda_n$ is large, thereby providing variable selection.

The proposals of \cite{henderson2014} and \cite{wu2014} are slightly more involved than \eqref{eqn::np}: an extra $\ell_2$-penalty is applied to the $\theta_{jk}$'s in \eqref{eqn::np_objective} in \cite{henderson2014}, whereas in \cite{wu2014} \eqref{eqn::np_objective} is followed by tuning parameter selection using Bayesian information criterion (\textsc{bic}), an adaptive group lasso regression, and a regular lasso.  We refer the reader to \cite{henderson2014} and \cite{wu2014}  for further details.

\section{PROPOSED APPROACH}\label{sec::ourmethod}

We consider the problem of model selection in high-dimensional ODEs. 
As in \cite{henderson2014} and \cite{wu2014}, we assume an additive ODE model \eqref{eqn::additivity}.
We use a finite basis $\psi(\cdot)$ to approximate the additive components $f_{jk}$ as in \eqref{eqn::fjk_approx}, leading to an ODE system that is linear in the unknown parameters \eqref{eqn::ODE_basis_finite}.
Following the example of \cite{dattner2015}, we exploit this linearity by integrating both sides of \eqref{eqn::ODE_basis_finite}, which yields 
\begin{equation}\label{eqn::ODE_integrated}\small 
X_j(t)=X_j(0) + \theta_{j0} t +\sum_{k=1}^p \Psi_k(t)^{\T} \theta_{jk} + \sum_{k=1}^p\int_0^t \delta_{jk}(X_k(u)) \, du, 
\end{equation}
where $\Psi_k(t)$ denotes the  integrated basis such that
\begin{equation}\label{eqn::Psi}
\Psi_k(t)= (\Psi_{k1}(t), \ldots, \Psi_{kM}(t))^{\T} = \int_0^t \psi(X_k(u)) \, du, \ k =1, \ldots, p,
\end{equation}
and $\Psi_0(t)=t$. 
Our method, called \textit{Graph Reconstruction via Additive Differential Equations} (GRADE), then solves the following problem for $j=1, \ldots, p$:
\begin{subequations}\label{eqn::us}
\begin{equation}
\small
\begin{split}
\hat{\theta}_j= \underset{C_{j0} \in \mathbb{R}, \theta_{j0} \in \mathbb{R}, \ \theta_{j1},\ldots, \theta_{jp} \in \mathbb{R}^{M}}{\arg \min} \ &  \frac{1}{2n} \sum_{i=1}^n \left\{Y_{ij} - C_{j0}- \theta_{j0}\hat{\Psi}_0(t_i) -\sum_{k=1}^p  \theta_{jk}^{\T} \hat{\Psi}_{k}(t_i)\right\}^2 \\
&+  \lambda_{n,j} \sum_{k=1}^p \left[ \frac{1}{n} \sum_{i=1}^n  \big\{\theta_{jk}^{\T}\hat{\Psi}_{k}(t_i)\big\}^2 \right]^{1/2},  
\end{split} \label{eqn::us_objective} 
\end{equation}
where
\begin{equation}
 \hat{X}(\cdot;h) = \underset{{Z}(\cdot) \in \mathcal{X}(h)}{\arg \min} \sum_{i=1}^n \|Y_i- Z(t_i)\|_2^2, \label{eqn::us_initial}\\
\end{equation}
and
\begin{equation}
 \hat{\Psi}_0(t)=t; \ \hat{\Psi}_k(t) = \int_0^t \psi(\hat{X}_k(u;h)) \, du, \ k=1, \ldots, p. \label{eqn::Psihat}
\end{equation}
\end{subequations}
In \eqref{eqn::us_objective}, $\lambda_{n,j}$ is a non-negative sparsity-inducing tuning parameter.
We may sometimes use $\lambda_{n,j} \equiv \lambda_n$ for $j=1,\ldots,p$ for simplicity. 
If the true function $f_{jk}^*$ in \eqref{eqn::additivity} is non-zero, we say that the $k$th variable $X^*_k$ is a true regulator of $X_j^*$. We let $S_j\equiv \{k: \|f^*_{jk}\|_2 \neq 0, k=1, \ldots, p\}$ denote the set of true regulators. We let the estimated index set of regulators be $\hat{S}_j\equiv \{k: \big\|\hat{\theta}_{jk}\big\|_2 \neq 0, k=1, \ldots, p\}$. We then  reconstruct the network using $\hat{S}_j, j=1,\ldots, p$.

Both \eqref{eqn::us_objective} and \eqref{eqn::us_initial}  can be implemented efficiently using existing software (e.g., \citealp{locfit, grplasso}). In our theoretical analysis in Section~\ref{sec::theory}, we use local polynomial regression to obtain the smoothing estimate in \eqref{eqn::us_initial}.  We  use generalized cross-validation (\textsc{gcv}) on the loss \eqref{eqn::us_initial} to select the smoothing tuning parameter $h$. We use \textsc{bic} to select the number of bases $M$ for $\psi$ and $\hat{\Psi}$ in \eqref{eqn::Psihat}, and the sparsity tuning parameter $\lambda_n$ in \eqref{eqn::us_objective}. 

In some studies, time-course data is collected from multiple samples, or experiments. Let $R$ denote the total number of experiments, and $Y^{(r)}$  the observations in the $r$th experiment. We assume that the same ODE system \eqref{eqn::ODE_basis_finite} applies across all experiments with the same true parameter $\theta^*_{jk}$. We allow a different set of initial values for each experiment. Assume that each experiment consists of measurements on the same set of time points. This leads us to modify \eqref{eqn::us} as follows:
\begin{equation}
\footnotesize
\label{eqn::multiple}
\begin{aligned}
\hat{\theta}_j= \underset{C_{j0}^{(r)} \in \mathbb{R}, \theta_{j0} \in \mathbb{R}, \ \theta_{j1},\ldots, \theta_{jp} \in \mathbb{R}^{M}}{\arg \min} & \frac{1}{2Rn} \sum_{r=1}^R \sum_{i=1}^n  \left\{Y_{ij}^{(r)} -C^{(r)}_{j0}- \theta_{j0} \hat{\Psi}_0(t_i)- \sum_{k=1}^p  \theta^{\T}_{jk}\hat{\Psi}^{(r)}_k(t_i) \right\}^2\\
 & + \lambda_n \sum_{k=1}^p \left[ \frac{1}{Rn} \sum_{r=1}^R \sum_{i=1}^n  \big\{\theta_{jk}^{\T}\hat{\Psi}^{(r)}_{k}(t_i)\big\}^2 \right]^{1/2},
\end{aligned}
\end{equation}
where
\begin{equation*}
\hat{X}^{(r)}(\cdot;h) = \underset{{Z}(\cdot) \in \mathcal{X}(h)}{\arg \min} \sum_{i=1}^n \|Y_i^{(r)}- Z(t_i)\|_2^2, \ r=1,\ldots, R,
\end{equation*}
\begin{equation*}
\hat{\Psi}_0(t)=t; \ \hat{\Psi}^{(r)}_k(t) = \int_0^t \psi\big(\hat{X}^{(r)}_k(u;h)\big) \, du, \ k=1, \ldots, p.
\end{equation*}
In Sections~\ref{sec::theory}, \ref{sec::sparse}, and \ref{sec::linear}, we will assume that only one experiment is available, so that our proposal takes the form \eqref{eqn::us}. In Sections~\ref{sec::LV} and \ref{sec::RDA}, we will apply our proposal to data from multiple experiments using \eqref{eqn::multiple}.

\begin{remark}
 To facilitate the comparison of GRADE \eqref{eqn::us} with other methods, we introduce an intermediate variable,  
\begin{equation}
\label{eqn::intermeidiate}
\tilde{X}_j(t;h,\theta) \equiv  C_{j0}+\theta_{j0}t+ \sum_{k=1}^p   \theta^{\T}_{jk} \hat{\Psi}_k(t),
\end{equation} 
following from \eqref{eqn::ODE_integrated}. Plugging \eqref{eqn::intermeidiate} into the loss function in \eqref{eqn::us_objective} yields $\sum_{i=1}^n \big\{Y_{ij} - \tilde{X}_j(t_i;h, \theta)\big\}^2$. In the gold standard \eqref{eqn::gold}, the ODE system \eqref{eqn::ODE_model} is strictly satisfied due to the constraint in \eqref{eqn::gold_constraint}. In the two-step procedure \eqref{eqn::two_objective} and \eqref{eqn::np_objective}, the smoothing estimate $\hat{X}(\cdot;h)$ does not satisfy \eqref{eqn::ODE_model}. GRADE stands in between:   the initial estimate $\hat{X}(\cdot;h)$  in \eqref{eqn::us_initial} is solely based on the observations, while the intermediate estimate $\tilde{X}(\cdot;h,\theta)$ is calculated by plugging $\hat{X}(\cdot;h)$ into the additive ODE \eqref{eqn::ODE_basis_finite}. 
\end{remark}

\section{THEORETICAL PROPERTIES}\label{sec::theory}

In this section, we establish variable selection consistency of the GRADE estimator \eqref{eqn::us}. Technical proofs of the statements in this section are available in  Section~\ref{sec::proofs} in the supplementary material.  We use $s_j$ to denote the cardinality  of $S_j$, and set $s=\max_j\{s_j\}$. For ease of presentation, we let $S_j^0=\{0\} \cup S_j$, so that $\Psi_{S^0_j}(t)=\big(\Psi_0(t), \Psi_{S_j}^{\T}(t)\big)^{\T}= \big(t,\Psi_{S_j}^{\T}(t)\big)^{\T}$ is an $(s_j M +1)$-vector. 

The proposed method \eqref{eqn::us} differs from the standard sparse additive model \citep{ravikumar2009} in that the regressors $\hat{\Psi}_k(t)$ in \eqref{eqn::Psihat} are estimated from smoothing estimates $\hat{X}(\cdot;h)$ \eqref{eqn::us_initial} instead of the true trajectories $X^*$ in \eqref{eqn::noisy}. We use local polynomial regression to compute $\hat{X}(\cdot;h)$ in \eqref{eqn::us_initial} (see e.g., Equation 1.67 of \citealp{tsybakov2009} for details on parameterization).  To establish variable selection consistency, it is necessary to obtain a bound for the difference between $\hat{X}(\cdot;h)$ and $X^*$. This is addressed in Theorem~\ref{thm::localP}. Using the bound in Theorem~\ref{thm::localP}, we then establish variable selection consistency of the estimator in \eqref{eqn::us} for high-dimensional ODEs in Theorem~\ref{thm::main}.

In this study, we assume that the measurement errors in \eqref{eqn::noisy} are normally distributed. Generalizations to bounded or sub-Gaussian errors are straightforward. 

\begin{assumption}\label{asmp::errors}
The measurement errors in \eqref{eqn::noisy} are independent, and $\epsilon_{ij} \sim N(0,\sigma^2), i=1,\ldots,n, j=1,\ldots, p$. 
\end{assumption}

We also require the true trajectories $X_j^*$ in \eqref{eqn::noisy} to be smooth.  

\begin{assumption}\label{asmp::smoothness}
Assume that the solutions $X_j^*, 1\leq j \leq p,$ belong to a H\"{o}lder class $\Sigma(\beta_1,L_1)$, where $\beta_1 \geq 3$.
\end{assumption}

In addition, we need some regularity assumptions to hold for the smoothing estimation \eqref{eqn::us_initial}. 
These assumptions are common and not crucial to this study, and are hence deferred to Section~\ref{sec::proof_local} in the supplementary material (or see Section~1.6.1 in \citealp{tsybakov2009}). We arrive at the following concentration inequality for $\vertiii{\hat{X}-X^*}$.

\begin{theorem}\label{thm::localP}
Suppose that Assumptions~\ref{asmp::errors}--\ref{asmp::smoothness} and \ref{asmp::LP1}--\ref{asmp::LP3} in the supplementary material are satisfied.
Let $\hat{X}_j$ in \eqref{eqn::us_initial} be the local polynomial regression estimator of order $\ell=\lfloor \beta_1 \rfloor$ with bandwidth 
\begin{equation}
h_n\propto n^{(\alpha-1)/(2\beta_1+1)}
\end{equation}
for some positive $\alpha<1$. Then, for each $j=1, \ldots, p$, 
\begin{equation}
\label{eqn::Xhat_single}
 \vertiii{\hat{X}_j- X_j^*}^2 \leq C_2 n^{\frac{2\beta_1 }{2\beta_1+1}(\alpha-1)} 
\end{equation}
holds with probability at least  $1-2\exp\big\{- n^{\alpha}/(2C_3\sigma^2)\big\}$, for some constants $C_2$ and $C_3$.
\end{theorem}
The concentration inequality in Theorem~\ref{thm::localP} is derived using concentration bounds for Gaussian errors  \citep{boucheron2013}. Using Theorem~\ref{thm::localP}, we see that the bound \eqref{eqn::Xhat_single} holds uniformly for $j=1,\ldots,p$ with probability at least $ 1-2p\exp\big\{- n^{\alpha}/(2C_3\sigma^2)\big\}$. The bound in Theorem~\ref{thm::localP} thus holds uniformly for $j=1,\ldots, p$ with probability converging to unity if $p=o\big(\exp\big\{n^{\alpha}/(2C_3\sigma^2)\big\}\big)$.

For the methods outlined in \eqref{eqn::np} \citep{henderson2014,wu2014}, variable selection consistency depends on the convergence of $\vertiii{\hat{X}'-(X^{*})'}$ and $\vertiii{\hat{X}-X^*}$. In contrast, our method depends only on the convergence rate of $\vertiii{\hat{X}-X^*}$. It is known that the convergence of $\vertiii{\hat{X}'-(X^{*})'}$ is slower than that of $\vertiii{\hat{X}-X^*}$, see e.g. \cite{gugushvili2012}. As a result, the rate of convergence of $\hat{\theta}_{jk}$ from \eqref{eqn::np} is slower than that of our proposed method \eqref{eqn::us}.
 
In order to establish the main result, we need the following additional assumptions. Recall the definition of $\Psi_j(t)$ from \eqref{eqn::Psi}; for convenience, we suppress the dependence of $\Psi(t)$ on $t$ in what follows. 

\begin{assumption}\label{asmp::additivity}
For $j=1,\ldots,p$, $(X^*_j)'$ is an additive function of $X^*_k$, $k=1,\ldots, p$. In other words,
\begin{equation}
\big(X^*_j \big)'(t) = \theta^*_{j0}  + \sum_{k=1}^p f^*_{jk}\big(X^*_k(t)\big),  \quad \theta^*_{j0} \in \mathbb{R}, \ j=1,\ldots, p,
\end{equation}
where $ \int_0^1 f^*_{jk}\big(X^*_k(t)\big) dt =0$ for all $j,k$. 
Furthermore, the functions $f_{jk}^* \ (1 \leq j,k \leq p)$  belong to a Sobolev class $W(\beta_2,L_2)$ on a finite interval with $\beta_2 \geq 3$.
\end{assumption}

\begin{assumption}\label{asmp::coherence_pop}
The eigenvalues of  $\int_0^1 {\Psi}_{S_j^0} {\Psi}_{S_j^0}^{\T} \, dt$ are bounded from above by $C_{\max}$ and  bounded from below by a positive number $C_{\min}$, and for $k \notin S_j^0$, the eigenvalues of  $\int_0^1 {\Psi}_{k} {\Psi}_{k}^{\T} \, dt$ are bounded from below by $C_{\min}$. In other words,
\begin{equation}
0< C_{\min} \leq \Lambda_{\min} \left( \int_0^1 {\Psi}_{S_j^0} {\Psi}_{S_j^0}^{\T} \, dt\right) \leq \Lambda_{\max} \left( \int_0^1 {\Psi}_{S_j^0} {\Psi}_{S_j^0}^{\T} \, dt\right) \leq C_{\max}, 
\end{equation}
and 
\begin{equation}
C_{\min} \leq \Lambda_{\min} \left( \int_0^1 {\Psi}_{k} {\Psi}_{k}^{\T} \, dt\right), \quad \text{for} \quad k \notin S_j^0.
\end{equation}
\end{assumption}

\begin{assumption}\label{asmp::irrepresentability_pop}
Assume that 
\begin{equation}
\max_{k \notin S_j^{0}} \left\|  \left( \int_0^1 {\Psi}_{k} {\Psi}_{S_j^0}^{\T} \, dt \right) \left( \int_0^1 {\Psi}_{S_j^0} {\Psi}_{S_j^0}^{\T} \, dt \right)^{-1}   \right\|_2 \leq  \xi.
\end{equation}
\end{assumption}

The first part of Assumption~\ref{asmp::coherence_pop} ensures identifiability among the $s_j+1$ elements in the set $\{t, X^*_{S_j}\}$, and the second part ensures that ${\Psi}_{k}$ is non-degenerate for $k \notin S_j^0$. 
Assumption~\ref{asmp::irrepresentability_pop} restricts the association between the elements in the set $\{t, X^*_{S_j}\}$ and the elements in the set $X^*_{S_j^c}$. 
Note that in order for the parameters in an additive model such as \eqref{eqn::ODE_basis_finite} to be identifiable, there must be  no concurvity among the variables \citep{buja1989}. This is guaranteed by Assumptions~\ref{asmp::coherence_pop} and \ref{asmp::irrepresentability_pop}, which appear often in the literature of lasso regression \citep{meinshausen2006,  zhao2006, ravikumar2009,  wainwright2009,  lee2013}. We refer the readers to \cite{miao2011} for a detailed discussion of the identifiability of the parameters in an ODE model. 

The next assumption characterizes the relationships between the quantities in Assumptions~\ref{asmp::coherence_pop} and \ref{asmp::irrepresentability_pop} and the sparsity tuning parameter $\lambda_n$ in \eqref{eqn::us_objective}. Similar assumptions have been made in lasso-type regression \citep{meinshausen2006,  zhao2006, ravikumar2009,  wainwright2009, lee2013}.  

\begin{assumption}\label{asmp::thetamin}
 Assume that
\begin{equation*}
f_{\min}>  \lambda_n \frac{4\sqrt{2sC_{\max}}}{C_{\min}} \quad  \text{and} \quad 
\xi <  \frac{1}{4}\sqrt{\frac{C_{\min}}{s C_{\max}}},
\end{equation*}
where $f_{\min} \equiv \min_{k\in S_j} \left\{\int_0^1 \big[f_{jk}^*(X_k^*(t))\big]^2 dt \right\}^{1/2}$ is the minimum regulatory effect.
\end{assumption}
 
Furthermore, we impose some regularity conditions on the bases $\psi(\cdot)$; these are deferred to Assumption~\ref{asmp::psibound} in the supplementary material. 

We arrive at the following theorem. 

\begin{theorem}\label{thm::main}
Suppose that Assumptions~\ref{asmp::errors}--\ref{asmp::thetamin} and \ref{asmp::LP1}--\ref{asmp::psibound} in the supplementary material hold, and let 
\begin{equation*}
{h}_n \propto \  n^{(\alpha-1)/(2\beta_1+1)}, \quad
M \propto \    n^{\frac{2\beta_1(1-\alpha)}{(2\beta_1+1)(2\beta_2+1)}}, \quad
\lambda_n \propto \  n^{-\frac{\beta_1(2\beta_2-1) (1-\alpha)}{(2\beta_1+1)(2\beta_2+1)}+2\gamma},  
\end{equation*}
where $0<\alpha<1$, $0<\gamma<H_1(\beta_1,\beta_2,\alpha)$, and $H_1(\beta_1,\beta_2,\alpha)$ is a constant that depends only on $\beta_1, \beta_2$ and $\alpha$. Then as $n$ increases,  
the proposed procedure \eqref{eqn::us} correctly recovers the true graph, i.e., $\hat{S}_j=S_j$ for all $j=1,\ldots, p$, with probability converging to $1$, if $s=O(n^{\gamma})$ and $pn\exp(-C_4 n^{\alpha}/\sigma^2)=o(1)$ for some constant $C_4$. 
\end{theorem}

Because the regressors $\hat{\Psi}$ are estimated, establishing variable selection consistency requires extra attention. To prove Theorem~\ref{thm::main}, we must first establish variable selection consistency of group lasso regression with errors in variables. This  generalizes the recent work on errors in variables for lasso regression \citep{loh2012}. Theorem~\ref{thm::main} ensures that the proposed method is able to recover the true graph exactly, given sufficiently dense observations in a finite time interval if the graph is sparse. The number of variables in the system can grow exponentially fast with respect to $n$, which means that the result holds for the ``large $p$, small $n$" scenario. 

Theorem~\ref{thm::main} does not provide us with practical guidance for selecting the bandwidth $h_n$ for the local polynomial regression estimator $\hat{X}_j$. The next result mirrors   Theorem~\ref{thm::main} for the bandwidths selected by cross-validation or \textsc{gcv}, which  converge to $h_n \propto n^{-1/(2\beta_1+1)}$ asymptotically (see \citealp{xia2002,tsybakov2009} for details). 

\begin{proposition}\label{prop::main}
Suppose that Assumptions~\ref{asmp::errors}--\ref{asmp::thetamin} and \ref{asmp::LP1}--\ref{asmp::psibound} in the supplementary material hold, and let 
\begin{equation*}
{h}_n \propto \  n^{-1/(2\beta_1+1)}, \quad
M \propto \   n^{\frac{1}{2\beta_2+1}(\frac{2\beta_1}{2\beta_1+1}-\alpha)}, \quad \text{and} \quad 
\lambda_n \propto \  n^{-\frac{2\beta_2-1}{4\beta_2+2}(\frac{2\beta_1}{2\beta_1+1}-\alpha)+2\gamma},  
\end{equation*}
where $0<\alpha<\frac{2\beta_1}{2\beta_1+1}$, $0<\gamma<H_2(\beta_1,\beta_2,\alpha)$, and $H_2(\beta_1,\beta_2,\alpha)$ is a constant that depends only on $\beta_1, \beta_2$ and $\alpha$. Then as $n$ increases, the proposed procedure \eqref{eqn::us} correctly recovers the true graph, i.e., $\hat{S}_j=S_j$ for all $j=1,\ldots, p$, with probability converging to $1$, if  $s=O(n^{\gamma})$ and $pn\exp(-C_4 n^{\alpha}/\sigma^2)=o(1)$ for some constant $C_4$. 
\end{proposition}

We note that selecting the values of $M$ and $\lambda_n$ that yield the rate specified in Proposition~\ref{prop::main} is challenging in practice.  
The rate of convergence of the sparsity tuning parameter $\lambda_n$ is slower in Proposition~\ref{prop::main} compared to  Theorem~\ref{thm::main}. This results in an increase in  the minimum regulatory effect $f_{\min}$ because of the relation between $f_{\min}$ and $\lambda_n$ in Assumption~\ref{asmp::thetamin}. 

\section{NUMERICAL EXPERIMENTS}\label{sec::simulation}

We study the empirical performance of our proposal in three different scenarios in the following subsections.
In what follows, given a set of initial conditions and a system of ODEs, numerical solutions of the ODEs are obtained using the Euler method with step size $0.001$.
Observations are drawn from the solutions at an evenly-spaced time grid $\{iT/n; i=1,\ldots, n\}$ with independent $N(0,1)$ measurement errors, unless specified otherwise. 
To facilitate the comparison of GRADE with other methods, we fit the smoothing estimates $\hat{X}$ in \eqref{eqn::us_initial} using smoothing splines with bandwidth chosen by \textsc{gcv}. 
We use cubic splines with two internal knots as the basis functions in \eqref{eqn::Psihat} in Sections~\ref{sec::sparse} and \ref{sec::LV}. Linear basis functions are used in Section~\ref{sec::linear}. 
The integral $\hat{\Psi}_k(t)=\int_0^t \psi\big(\hat{X}_k(u;{h})\big) \, du$ in \eqref{eqn::Psihat} is calculated numerically with step size $0.01$.

\subsection{Variable selection in additive ODEs}\label{sec::sparse}

In this simulation, we compare GRADE with NeRDS \citep{henderson2014} and SA-ODE \citep{wu2014} described in \eqref{eqn::np}. 
We consider the following system of additive ODEs, for $k=1,\ldots, 5$:  
\begin{equation}
\small
\begin{cases}
& X_{2k-1}'(t)= \theta_{2k-1,0}+\psi(X_{2k-1}(t))^{\T}\theta_{2k-1,2k-1} + \psi(X_{2k}(t))^{\T}\theta_{2k-1,2k}\\
& X_{2k}'(t)= \theta_{2k,0}+\psi(X_{2k-1}(t))^{\T}\theta_{2k,2k-1} + \psi(X_{2k}(t))^{\T}\theta_{2k,2k} 
\end{cases}, t \in [0,20],
\label{eqn::FHN}
\end{equation}
where $\psi(x)=(x,x^2,x^3)^{\T}$ is the cubic monomial basis.  
The parameters and initial conditions are chosen so that the solution trajectories are identifiable under an additive model \citep{buja1989}. 
Detailed specification of \eqref{eqn::FHN} can be found in Section~\ref{sec::data_appendix} of the supplementary material.

After generating data according to \eqref{eqn::FHN} and introducing noise, we apply GRADE, NeRDS, and SA-ODE to recover the directed graph encoded in \eqref{eqn::FHN}.  
Both NeRDS and SA-ODE are implemented using code provided by the authors. NeRDS and SA-ODE use smoothing splines to estimate $\hat{X}$ and $\hat{X}'$ in \eqref{eqn::np_constraint}, and cubic splines with two internal knots as the basis $\psi$ in \eqref{eqn::np_objective}. As mentioned briefly in Section~\ref{sec::method}, NeRDS applies an additional smoothing penalty which amounts to an $\ell_2$ penalty on $\theta_{jk}$ in \eqref{eqn::np_objective}, controlled by a parameter selected using \textsc{gcv} \citep{henderson2014}. We apply GRADE using the same smoothing estimates and basis functions as NeRDS and SA-ODE.
To facilitate a direct comparison to NeRDS, we apply GRADE both with and without an additional $\ell_2$-type penalty on the $\theta_{jk}$'s in \eqref{eqn::us_objective}. 
We apply all methods for a range of values of the sparsity-inducing tuning parameter (e.g., $\lambda_n$ in \eqref{eqn::us_objective}), in order to yield a  recovery curve of varying sparsity.

We summarize the simulation results in Figure~\ref{fig::comparison}, where the numbers of true edges selected are displayed against the total numbers of selected edges over a range of sparsity tuning parameters.   
We see that GRADE outperforms the other two methods, which corroborates our theoretical findings in Section~\ref{sec::theory} that our proposed method is more efficient than methods such as NeRDS and SA-ODE which involve derivative estimation (see e.g., comments below Theorem 1).

\begin{figure}[ht]
\centering
   \centering
   \subfigure{\includegraphics[scale=1]{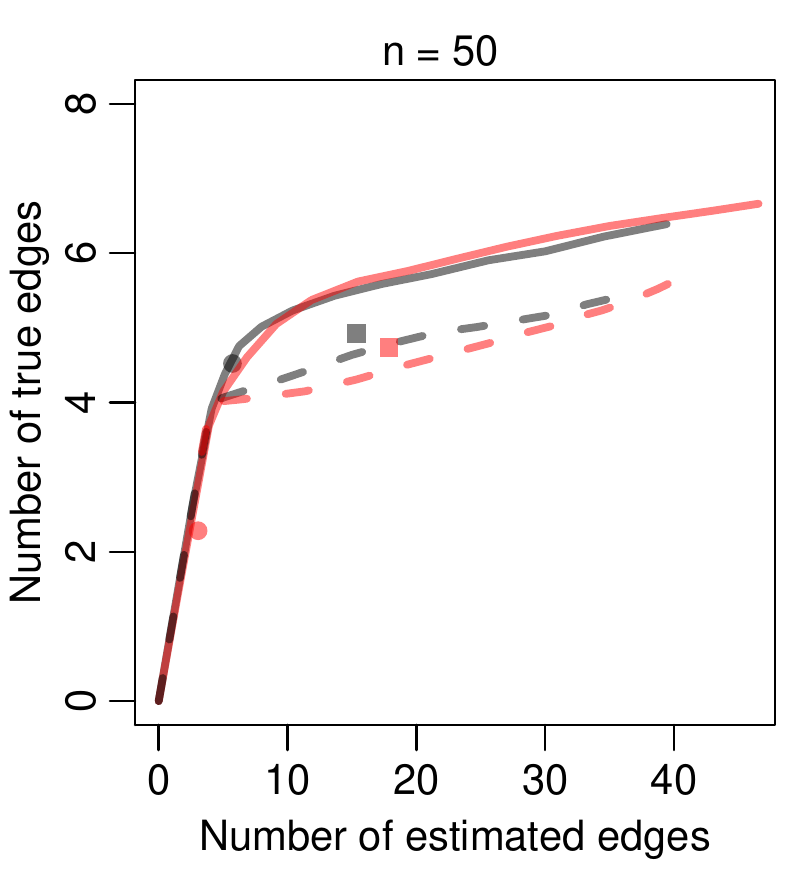}}
   \subfigure{\includegraphics[scale=1]{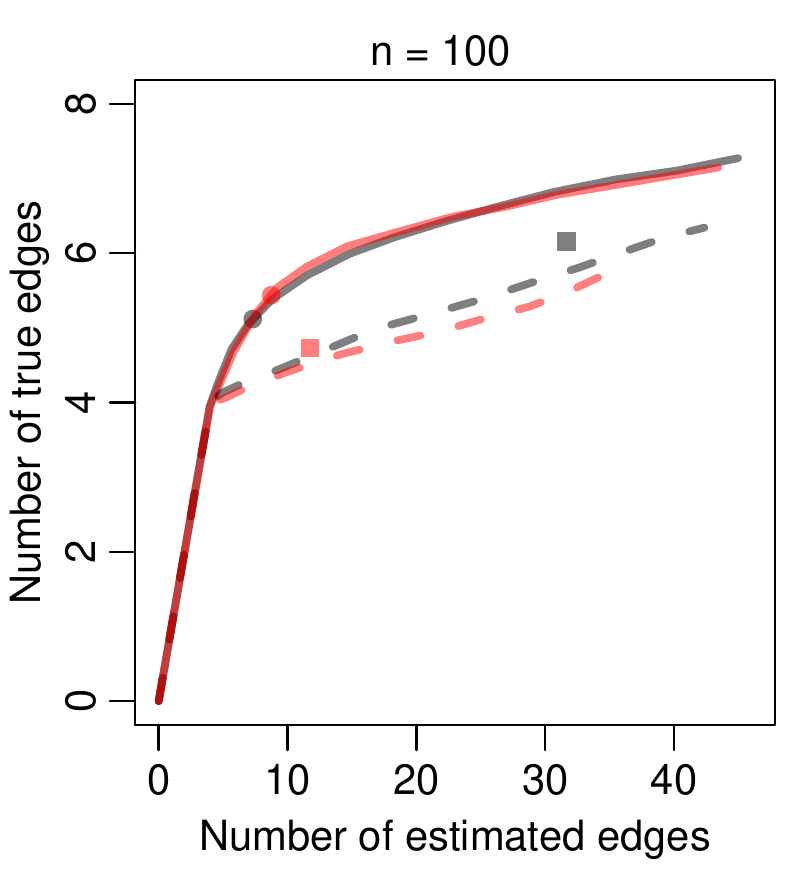}}
\caption{Performance of network recovery methods on the system of additive ODEs in \eqref{eqn::FHN}, averaged over 400 simulations.  The four curves represent SA-ODE (\protect\includegraphics[height=0.5em]{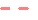}), NeRDS (\protect\includegraphics[height=0.5em]{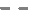}),  and GRADE without (\protect\includegraphics[height=0.5em]{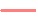}) and with (\protect\includegraphics[height=0.5em]{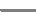}) the additional smoothing penalty in \eqref{eqn::us_objective} used by NeRDS.  Each point on the curves corresponds to average performance for a given sparsity tuning parameter $\lambda_n$  in \eqref{eqn::np_objective} or  \eqref{eqn::us_objective}. The symbols indicate the sparsity tuning parameter $\lambda_n$ selected using \textsc{bic} (SA-ODE, \protect\includegraphics[height=0.5em]{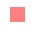}, and GRADE, \protect\includegraphics[height=0.5em]{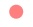} and \protect\includegraphics[height=0.5em]{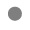}) or \textsc{gcv} (NeRDS, \protect\includegraphics[height=0.5em]{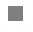}).} 
\label{fig::comparison}
\end{figure}

\subsection{Variable selection in linear ODEs}\label{sec::linear}

In this simulation, we compare GRADE to two recent proposals by \cite{brunel2014} and \cite{hall2014}.
Recall from Section~\ref{sec::two-step} that \cite{brunel2014} and \cite{hall2014} are proposed to estimate a few unknown parameters in an ODE system of known form.
Hence, we consider a simple linear ODE system, for $k=1,\ldots,4$, 
\begin{equation}\label{eqn::linearODE}
\begin{cases}
& X_{2k-1}'(t) = 2k \pi X_{2k}(t)\\
& X_{2k}'(t)= - 2k \pi X_{2k-1}(t)\\
\end{cases}, t \in [0,1].
\end{equation}
For each $k=1,\ldots,4$, we set the initial condition to be $(X_{2k-1}(0), X_{2k}(0))=(\sin(y_k), \cos(y_k))$ where $y_k\sim N(0,1)$.  
The solutions to \eqref{eqn::linearODE} take the form of sine and cosine functions of frequencies ranging from $2\pi$ to $8 \pi$. 
The graph corresponding to \eqref{eqn::linearODE} is sparse, with only eight directed edges out of $64$ possible edges. 
We fit the model 
\begin{equation}\label{eqn::linearmodel}
X'(t)=\Theta X(t) + C,
\end{equation}
where $\Theta$ is an unknown $8 \times 8$ matrix and $C$ is an $8$-vector.
We apply the method in \cite{brunel2014} using the code provided by the authors. 
We implement the method in \cite{hall2014} in R based on the authors' code in Fortran. 
Because the loss function in \cite{hall2014} is not convex, we use five sets of random initial values and report the best performance.
Since both \cite{brunel2014} and \cite{hall2014} yield dense estimates for $\Theta$ in \eqref{eqn::linearmodel}, in order to examine how well these methods recover the true graph, we threshold the estimates at a range of values in order to obtain a variable selection path.  
We apply GRADE using the linear basis function $\psi(x)=x$. 

Results are shown in Figure~\ref{fig::linear}. 
We can see that GRADE outperforms the methods in \cite{brunel2014} and \cite{hall2014}. 
This is likely due to the fact that GRADE exploits the sparsity of the true graph with a sparsity-inducing penalty.
In principle, \cite{brunel2014} and \cite{hall2014} could be generalized in order to include penalties on the parameters. 
We leave this to future research.

\begin{figure}[ht]
\centering
   \centering
   \includegraphics[scale=1]{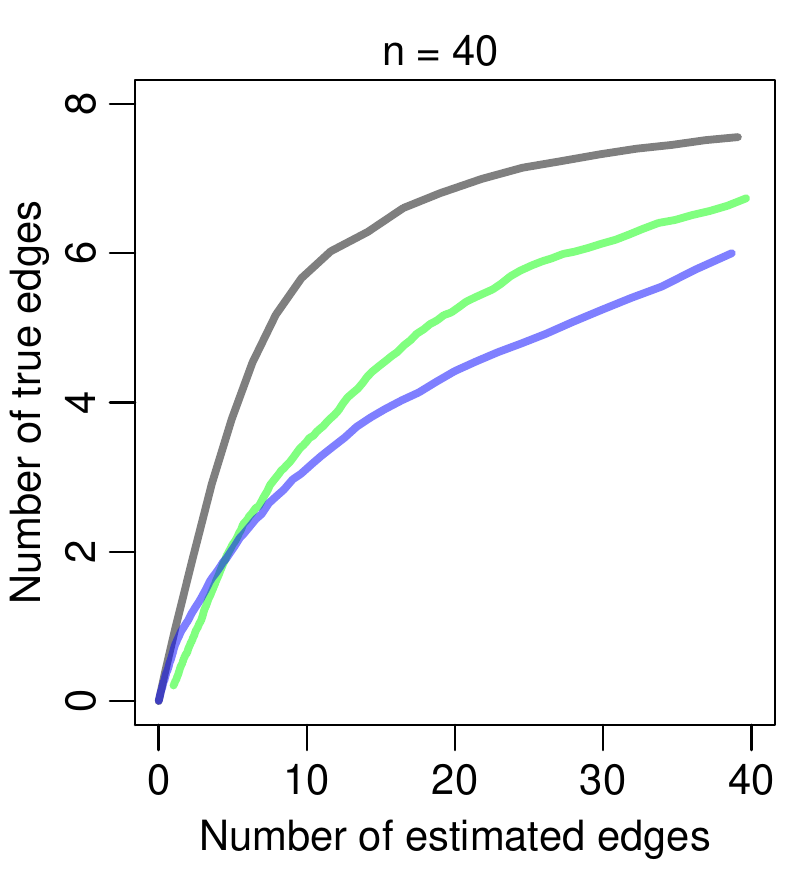} \includegraphics[scale=1]{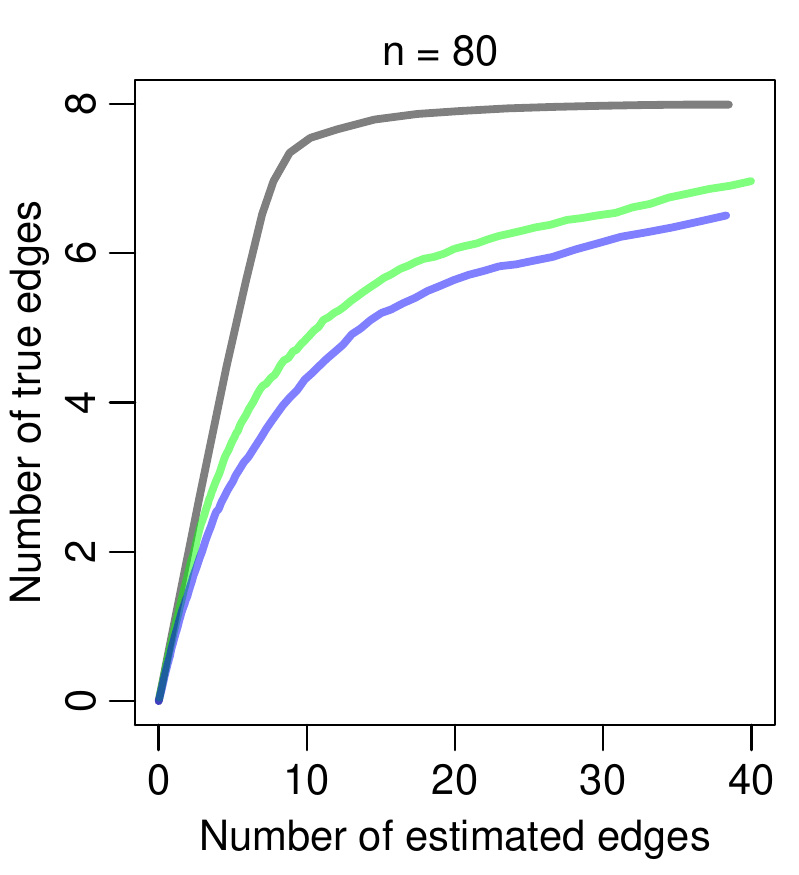}
\caption{Network recovery on the system of linear ODEs \eqref{eqn::linearODE}, averaged over 200 simulated data sets.  The three curves represent GRADE (\protect\includegraphics[height=0.5em]{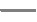}), \cite{hall2014} (\protect\includegraphics[height=0.5em]{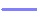}), \cite{brunel2014} (\protect\includegraphics[height=0.5em]{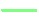}).}
\label{fig::linear}
\end{figure} 

\subsection{Robustness of GRADE to the additivity assumption}\label{sec::LV}

The GRADE method assumes that  the true underlying model is additive (Assumption~\ref{asmp::additivity}). 
However, in many systems, the additivity assumption is violated; for instance,  multiplicative effects may be present in gene regulatory networks \citep{ma2009}. 
In this subsection, we  investigate the performance of GRADE in a setting where the true model is non-additive.  
We consider the following system of ODEs, for $k=1,\ldots, 5$, 
\begin{equation}\label{eqn::LV}
\begin{cases}
& X_{2k-1}'(t) = f_{2k-1}\left(X_{2k-1}(t),X_{2k}(t)\right) \equiv 2 X_{2k-1}(t) - v X_{2k-1}(t)X_{2k}(t)\\
& X_{2k}'(t) = f_{2k}\left(X_{2k-1}(t),X_{2k}(t) \right) \equiv v X_{2k-1}(t)X_{2k}(t) - 2 X_{2k}(t)\\
\end{cases}, t \in [0, 5], 
\end{equation}
where $v$ is a positive constant. 
For each $k=1,\ldots,5$, the pair of equations \eqref{eqn::LV} is a special case of the Lotka-Volterra equations \citep{volterra1928}, which represent the dynamics between predators ($X_{2k}$) and prey ($X_{2k-1}$). 
The parameter $v$ defines the interaction between the two populations. 
For $v \neq 0$, both $X'_{2k-1}$ and $X'_{2k}$ are non-additive functions of $X_{2k-1}$ and $X_{2k}$.
We define two types of directed edges, where $\mathcal{E}_1\equiv\{ (X_{j}, X_{j}), j=1,\ldots, 10\}$ and $\mathcal{E}_2\equiv\{ (X_{2k-1}, X_{2k}), (X_{2k}, X_{2k-1}), \   k=1,\ldots, 5\}$ represent the self-edges and non-self-edges, respectively. Figure~\ref{fig::recoveryLV}(a) contains an illustration of the graph and edge types for each pair of equations. 
In what follows, we investigate how well GRADE recovers these two types of edges as we change the parameter $v$, i.e., as the additivity assumption is violated.   

Since measurement error is not essential to the current discussion, we generate data according to \eqref{eqn::LV} without adding noise.
To ensure that the trajectories are identifiable, we generate $R=2$ sets of random initial values drawn from $N_{10}(0, 2{I}_{10})$, where ${I}_{10}$ is a $10 \times 10$ identity matrix.  
In order to quantify the amount of signal in an edge that GRADE can  detect, we introduce the quantity
\begin{equation}\label{eqn::dfdxl2}
D_{j,k}(v)  = \mathbb{E} \left[ R \int_0^T \left\{ \frac{\partial f_{j} }{\partial X_{k}} \left(t; X(0)\right) \right\}^2  dt\right],
\end{equation} 
where the expectation is taken with respect to the random initial values $X(0)$ and $R$ is the number of initial values.
The measure $D_{j,k}$ in \eqref{eqn::dfdxl2} is a loose analogy to $\left\{\int_0^1 \big[f_{jk}^*(X_k^*(t))\big]^2 dt \right\}^{1/2}$ used in Assumption~\ref{asmp::thetamin}. 
Note that if no edge is present from $X_{k}$ to $X_{j}$, then $\partial f_j/\partial  X_k \equiv 0$ and hence  $D_{j,k}(v)=0$. 
One immediately notes that, as $R$ increases, the regulatory effect for a true edge increases proportionally to $R$, while the regulatory effect of a non-edge remains zero. 
For the self-edges in $\mathcal{E}_1$ and the non-self-edges in $\mathcal{E}_2$, we can define $D^{(1)}(v)$ and $D^{(2)}(v)$ as 
\begin{equation}\label{eqn::signals_twotypes}
D^{(1)}(v)= \underset{k=1,\ldots,10}{\min} D_{k,k}(v), \quad \text{and} \quad  
D^{(2)}(v)= \underset{k=1,\ldots,5}{\min} \{D_{2k-1,2k}(v),D_{2k,2k-1}(v)\},
\end{equation}
where we use the minimum because variable selection is limited by the minimum regulatory effect (see Assumption~\ref{asmp::thetamin}). 
With a slight abuse of definition, we refer to \eqref{eqn::signals_twotypes} as the minimum regulatory effects in a non-additive model.

We apply GRADE using the formulation in \eqref{eqn::multiple}. 
The sparsity parameter $\lambda$ is chosen so that there are $20$ directed edges in the estimated network. 
We record the number of estimated edges that are in $\mathcal{E}_1$ and $\mathcal{E}_2$. 
The edge recovery performance is shown in Figure~\ref{fig::recoveryLV}(b). 
In Figure~\ref{fig::recoveryLV}(c), we display the minimum regulatory effects defined in \eqref{eqn::signals_twotypes}.
Edge recovery and minimum regulatory effects show a similar trend as a function of $r$ in \eqref{eqn::LV}. 
This suggests that \eqref{eqn::signals_twotypes}, and thus \eqref{eqn::dfdxl2}, is a reasonable measure of the additive components of the regulatory effect of the edges. 
The slight deviation between the trends reflects the fact that the measure defined in \eqref{eqn::dfdxl2} is not a direct counterpart of $\left\{\int_0^1 \big[f_{jk}^*(X_k^*(t))\big]^2 dt \right\}^{1/2}$ in a non-additive model.
The edge recovery improves when a larger value of $R$ is used, though these results are omitted due to space constraints. 
Our results indicate that GRADE can recover the true graph even when the additivity assumption is violated, provided that the regulatory effects \eqref{eqn::dfdxl2} for the true edges are sufficiently large.

\begin{figure}[ht]
\centering
   \centering
   \subfigure{\includegraphics[scale=1]{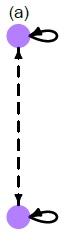}}\  \subfigure{\includegraphics[scale=1]{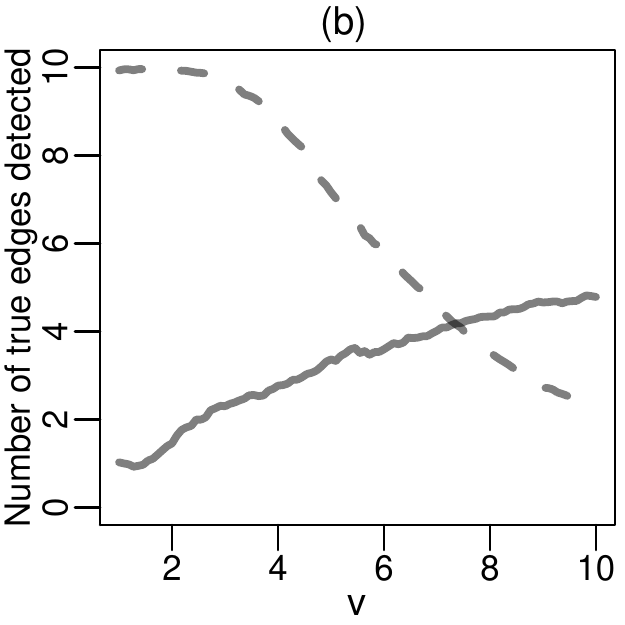}}\    \subfigure{\includegraphics[scale=1]{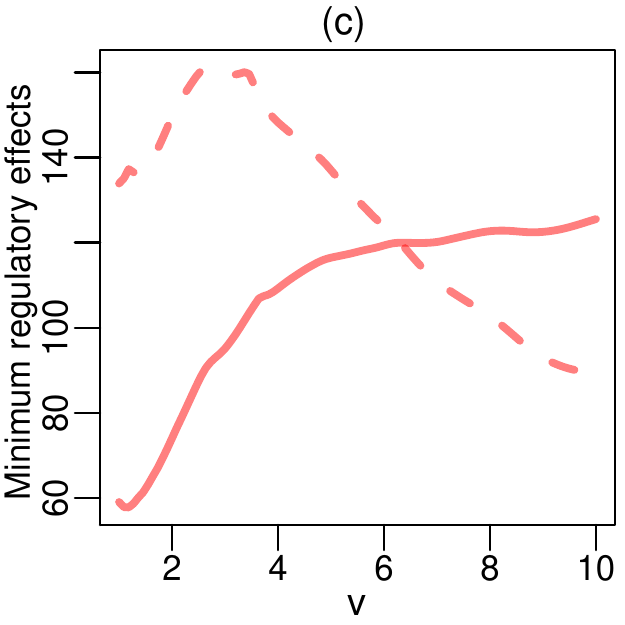}} 
\caption{(a): The graph encoded by a pair of Lotka-Volterra equations as given in \eqref{eqn::LV}. Self-edges  (\protect\includegraphics[height=0.5em]{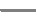}) and non-self-edges 
(\protect\includegraphics[height=0.5em]{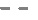}) are shown.  (b): Self-edge (\protect\includegraphics[height=0.5em]{./blacksolid-F3-colored.jpg}) and non-self-edge (\protect\includegraphics[height=0.5em]{./blackdashed-F3-colored.jpg}) recovery of GRADE, averaged over $200$ simulated data sets. (c): Minimum signals defined in \eqref{eqn::signals_twotypes}, for self-edges, $D^{(1)}(\cdot)$ (\protect\includegraphics[height=0.5em]{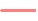}), and non-self-edges, $D^{(2)}(\cdot)$ (\protect\includegraphics[height=0.5em]{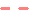}). } 
\label{fig::recoveryLV}
\end{figure}

\section{APPLICATIONS}\label{sec::RDA}

\subsection{Application to \emph{in silico} gene expression data}\label{sec::dream}

GeneNetWeaver (GNW) provides an \textit{in silico} benchmark for assessing the performance of network recovery methods \citep{schaffter2011}, and was used in the third DREAM challenge \citep{marbach2009}. 
GNW is based upon real gene regulatory networks of yeast and \textit{E. coli}. 
It extracts sub-networks from the yeast or \textit{E. coli} gene regulatory networks, and assigns a system of ODEs to the extracted network. 
This system of ODEs is non-additive, and includes unobserved variables \citep{marbach2010}. Therefore,  the assumptions of GRADE are violated in the GNW data.  
 
To mimic real-world laboratory experiments, GNW provides several data generation mechanisms. 
In this study, we consider data from the \textit{perturbation} experiments. 
The perturbation experiments are similar to the data generating mechanisms used in Section~\ref{sec::LV}, where initial conditions of the ODE system are perturbed in order to emulate the diversity of trajectories from multiple independent experiments.

We investigate ten networks from GNW that have been previously studied in \cite{henderson2014}, of which five have $10$ nodes and five have $100$ nodes. 
For each network,  GNW provides one set of noiseless gene expression data consisting of $R$ perturbation experiments where the trajectories are measured at $n=21$ evenly-spaced time points in $[0,1]$. 
Here $R=10$ for the five 10-node networks and $R=100$ for the five 100-node networks.
As in \cite{henderson2014},  we add independent $N\big(0,0.025^2\big)$ measurement errors to the data at each timepoint. 

We apply NeRDS as described in \cite{henderson2014}.  
We apply GRADE using the formulation \eqref{eqn::multiple} to handle observations from multiple experiments, with the smoothing estimates $\hat{X}$ in \eqref{eqn::us_initial} fit using smoothing splines with bandwidth chosen by \textsc{gcv},  and using cubic splines with two internal knots as the basis functions in \eqref{eqn::Psihat}. 
The integral $\hat{\Psi}_k(t)=\int_0^t \psi\big(\hat{X}_k(u;{h})\big) \, du$ in \eqref{eqn::Psihat} is calculated numerically with step size $0.01$.
Finally, we apply an additional $\ell_2$-type penalty to the $\theta_{jk}$'s in \eqref{eqn::multiple} in order to match the setup of NeRDS. 
The tuning parameter for this penalty is set to be $0.1$.

Results are shown in Table~\ref{tab::AUC}. 
Recall that the data generating mechanism violates crucial assumptions for both NeRDS and GRADE. 
We see in Table~\ref{tab::AUC} that 
NeRDS outperforms GRADE in one network, while GRADE outperforms NeRDS in the other nine networks.   
This suggests that GRADE is a competitive exploratory tool for reconstructing gene regulatory networks. 

\begin{table}[ht]
\caption{Area Under ROC Curves for NeRDS and GRADE}
\label{tab::AUC}
\begin{center}
\begin{tabular}{ c | c | c | c | c }
  \hline
  & \multicolumn{2}{|c|}{$p=10$} & \multicolumn{2}{|c}{$p=100$} \\
 & NeRDS & GRADE & NeRDS & GRADE \\ 
  \hline
Ecoli1 & $0.450$ $(0.438, 0.462)$ & $\bm{0.545}$ $(0.534, 0.557)$ & $0.624$ $(0.622, 0.627)$ & $\bm{0.670}$ $(0.667, 0.673)$ \\ 
  Ecoli2 & $0.512$ $(0.502, 0.523)$ & $\bm{0.643}$ $(0.634, 0.653)$ & $0.637$ $(0.635, 0.640)$ & $\bm{0.653}$ $(0.650, 0.656)$ \\ 
  Yeast1 & $0.486$ $(0.476, 0.495)$ & $\bm{0.679}$ $(0.666, 0.691)$ & $0.610$ $(0.607, 0.612)$ & $\bm{0.636}$ $(0.635, 0.638)$ \\ 
  Yeast2 & $0.525$ $(0.518, 0.532)$ & $\bm{0.607}$ $(0.600, 0.613)$ & $0.568$ $(0.566, 0.569)$ & $\bm{0.584}$ $(0.582, 0.585)$ \\ 
  Yeast3 & $0.467$ $(0.460, 0.474)$ & $\bm{0.576}$ $(0.566, 0.587)$ & $\bm{0.617}$ $(0.616, 0.619)$ & $0.567$ $(0.566, 0.568)$ \\ 
   \hline
\end{tabular}
The average area under the curves and $90\%$ confidence intervals, over 100 simulated data sets. 
Networks and data generating mechanisms are described in Section~\ref{sec::dream}.
Boldface indicates the method with larger AUC. 
\end{center}
\end{table}

\subsection{Application to calcium imaging recordings}\label{sec::CI}

In this section, we consider the task of learning regulatory relationships among populations of neurons. 
We investigate the calcium imaging recording data from the Allen Brain Observatory project conducted by the Allen Institute for Brain Science\footnote{Website: \copyright 2016 Allen Institute for Brain Science. Allen Brain Observatory [Internet]. Available from: \texttt{http://observatory.brain-map.org}.}.
Here, we investigate one of the experiments in the project.
In this experiment, calcium fluorescence levels (a surrogate for neuronal activity) are recorded at $30$ Hz on a region of the primary visual cortex while the subject mouse is shown forty visual stimuli.
The forty visual stimuli are combinations of eight spatial orientations and five temporal frequencies.
Each stimulus lasts for two seconds and is repeated 15 times.
The recorded videos are processed by the Allen Institute to identify individual neurons.  
In this particular experiment, there are $575$ neurons. 
Each neuron's activity is defined as the average calcium fluorescence level of the pixels that it covers in the video.

It is known that the activities of individual neurons are noisy and sometimes misleading \citep{cunningham2014}.
As an alternative, neuronal populations can be studied (see e.g.,  Part Three of \citealp{gerstner2014}). 
We define 25 neuronal populations by dividing the recording region into a $5 \times 5$ grid, where each population contains roughly 20 neurons.  
We use GRADE to capture the functional connectivity among the $25$ neuronal populations. 
Note that functional connectivity is distinct from physical connectivity. 
Functional connectivity involves the relationships among neuronal populations that can be observed through neuron activities and may change across stimuli, whereas physical connectivity consists of synaptic interactions. 

We estimate the functional connectivity corresponding to three different but related stimuli, consisting of frequencies of $1$ Hz, $2$ Hz, and $4$ Hz, each at a spatial orientation of $90^\circ$.
For each stimulus, we have calcium fluorescence levels of the $p=25$ neuronal populations for each of $R=15$ repetitions.
Since each repetition spans two seconds and the calcium fluorescence is recorded at $30$ Hz, there are  $60$ timepoints per repetition. 
We apply GRADE using the formulation in \eqref{eqn::multiple} in order to reconstruct the functional connectivity under each of the three stimuli. 
We use smoothing splines with bandwidth $h$ selected with \textsc{gcv} in order to estimate $\hat{X}$ in \eqref{eqn::us_initial}, and use cubic splines with $4$ internal knots as the basis functions $\psi(\cdot)$ in \eqref{eqn::Psihat}. 
The sparsity parameter  ${\lambda}_{j,n}$ for each nodewise regression in \eqref{eqn::multiple} is selected using \textsc{bic} for each $j=1,\ldots, 25$.
For ease of visualization, we prefer a sparse network, and so we fit GRADE using tuning parameter values $\alpha ({\lambda}_{1,n},  \ldots, {\lambda}_{p,n})$, where the scalar $\alpha$ is selected so that each of the estimated networks contains approximately 25 edges.

Estimated functional connectivities are shown in Figure~\ref{fig::recoveryRDA}. 
We see that, in all three networks, the $24$th neuronal population regulates many other neuronal populations, indicating that this region may contain neurons that are sensitive to this spatial orientation.
Furthermore, we see that the adjacent connectivity networks in Figure~\ref{fig::recoveryRDA} are somewhat similar to each other, whereas the networks at $1$ Hz and $4$ Hz have few similarities. 
This agrees with the observation in neuroscience that neurons in the mouse primary visual cortex are responsive to a somewhat narrow range of temporal frequencies near their peak frequencies (see, e.g., \citealp{gao2010}).
 
\begin{figure}[ht]
\centering
   \centering
   \includegraphics[scale=1]{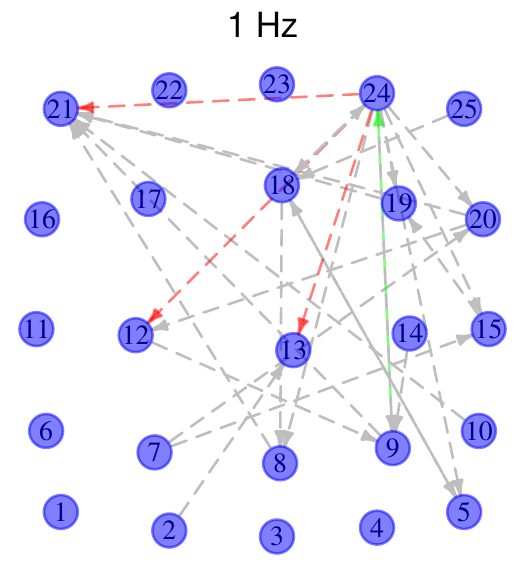}\ \includegraphics[scale=1]{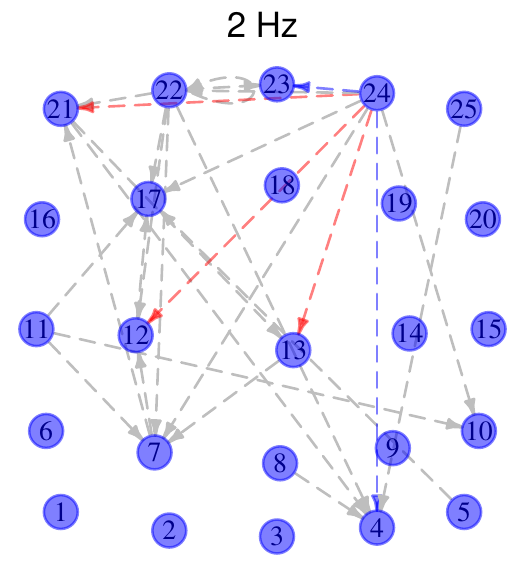}\ \includegraphics[scale=1]{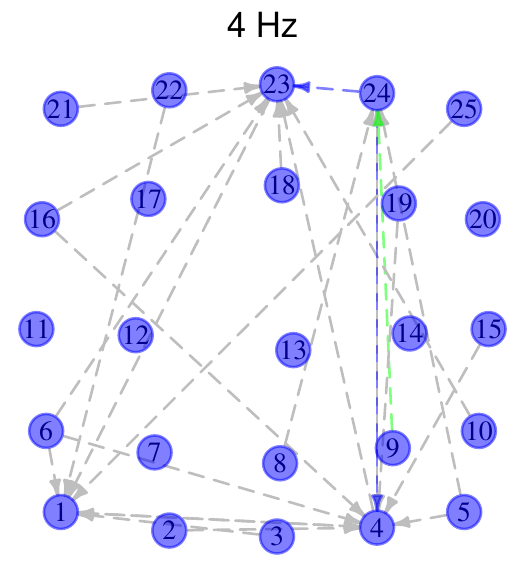}
\caption{Estimated functional connectivities among neuronal populations from the calcium imaging data described in Section~\ref{sec::CI}. 
Each node is positioned near the center of the neuronal population it represents, with jitter added for ease of display. 
The three red edges are shared between the estimated networks at $1$ Hz and $2$ Hz; the two blue edges are shared between estimated networks at $2$ Hz and $4$ Hz; the single green edge is shared between the estimated networks at $1$ Hz and $4$ Hz. For reference, given two Erd\"{o}s-R\`{e}nyi graphs consisting of $25$ nodes and $25$ edges, the probability of having three or more shared edges is $0.07$, and the probability of having two or more shared edges is $0.26$.}
\label{fig::recoveryRDA}
\end{figure} 

\section{DISCUSSION}\label{sec::discussion}

In this paper, we propose a new approach, GRADE, for estimating a system of high-dimensional additive ODEs. GRADE involves estimation of an integral rather than a derivative. We show that estimating the integral is superior to estimating the derivatives both theoretically and empirically.  We leave an extension of our work to non-additive ODEs to future research. 

In this paper, we have not addressed the issue of experimental design. Given a finite set of resources, one may choose to design an experiment to measure $n$ observations on a very dense time grid, or on a coarse time grid. Alternatively, one might choose to measure $n/R$ observations for $R$ distinct experiments from a single ODE system \eqref{eqn::ODE_model}, each with a different initial condition. 
This presents a trade-off that is especially interesting in the context of ODEs:
using a dense time grid improves the quality of the smoothing estimates $\hat{X}$, as seen in Sections~\ref{sec::sparse} and \ref{sec::linear}, while running multiple experiments enhance the identifiability of the true structure, as seen in Section~\ref{sec::LV}. 
We leave a more detailed treatment of these issues to future work.

 \bibliographystyle{chicago}
\bibliography{paper-ref}

\newpage 

\begin{center}
\Large  
 {\textbf{Supplementary Material for Network Reconstruction From High Dimensional Ordinary Differential Equations} }
\end{center}
\vspace{5mm}

\begin{appendices}

\section{PROOFS}\label{sec::proofs}
\subsection{Outline}

In this section, we prove Theorems~\ref{thm::localP} and \ref{thm::main} from Section~\ref{sec::theory} in the main paper. The remaining subsections are organized as follows. In Section~\ref{sec::proof_local}, we list the additional assumptions for Theorem~\ref{thm::localP}  in the main paper and give the proof of Theorem~\ref{thm::localP}  in the main paper. In Section~\ref{sec::group}, we prove a theorem  on variable selection consistency for group lasso regression with errors in variables, which itself is of independent interest. In  Section~\ref{sec::lemmas}, we introduce Assumption~\ref{asmp::psibound} on the bases $\psi(\cdot)$, and several technical lemmas that are useful in proving Theorem~\ref{thm::main}  in the main paper. In Section~\ref{sec::proof_main}, we finish the proof of Theorem~\ref{thm::main}  in the main paper. And in Section~\ref{sec::proof_prop}, we prove Proposition~\ref{prop::main} in the main paper. The proofs of the technical lemmas presented in Section~\ref{sec::lemmas} are provided in Section~\ref{sec::lemmas_proof}.

\subsection{Proof of Theorem~\ref{thm::localP}}\label{sec::proof_local}

In this section, we follow closely the notation in Section~1.6 of \cite{tsybakov2009}. We first present some necessary notation and assumptions. 
Denote the local polynomial estimator of degree $\ell$ as 
\begin{equation} \label{eqn::LPell}
\hat{X}(t;h)=\sum_{i=1}^n Y_i W_{ni}(t;h),
\end{equation}
where 
\begin{equation} \label{eqn::W_def}
 W_{ni}(t;h) =\frac{1}{nh} U^{\T}(0) B^{-1}_{nt} U\left(\frac{t_i-t}{h}\right) K\left( \frac{t_i-t}{h}\right),
\end{equation}

\begin{equation*}
\begin{aligned}
&B_{nt} =\frac{1}{nh}  \sum_{i=1}^n U\left(\frac{t_i-t}{h}\right) U^{\T}\left(\frac{t_i-t}{h}\right) K\left( \frac{t_i-t}{h}\right),\\
&U(u) = \big(1, u, u^2/2!, \ldots, u^{\ell}/{\ell}!\big)^{\T},
\end{aligned}
\end{equation*}
and $K(\cdot)$ is a kernel function. In \eqref{eqn::W_def}, $W_{ni}(t;h)$ is the weight for observation $Y_{i}$ in \eqref{eqn::LPell}, which satisfies 
\begin{equation}\label{eqn::W_sum}
\sum_{i=1}^n W_{ni}(t;h) =1. 
\end{equation}
See e.g., Proposition~1.12 in \cite{tsybakov2009}, for a rigorous proof of \eqref{eqn::W_sum}. We introduce the following assumptions on the kernel function $K(\cdot)$ and the time points $t_1,\ldots,t_n$. These assumptions are common in the study of local polynomial estimators (see e.g. \citealp{tsybakov2009}).

\begin{assumption}\label{asmp::LP1}
There exists a real number $\lambda_0>0$ and a positive integer $n_0$ such that the smallest eigenvalue $\Lambda_{\min}(B_{nt})$ of $B_{nt}$ satisfies 
$$\Lambda_{\min} (B_{nt}) \geq \lambda_0$$
for all $n\geq n_0$ and any $t \in [0,1]$.
\end{assumption}

\begin{assumption}\label{asmp::LP2}
The time points $t_1,\ldots, t_n$ are evenly-spaced on the interval $[0,1]$.
\end{assumption}

\begin{assumption}\label{asmp::LP3}
The kernel $K$ has compact support belonging to $[-1,1]$, and there exists a number $K_{\max} < \infty$ such that $|K(u)|\leq K_{\max}$, $\forall u \in \mathds{R}$.  
\end{assumption}

These assumptions lead to the following lemma (Lemma 1.3 in  \citealp{tsybakov2009}). 
\begin{lemma}\label{lmm::W}
Under Assumptions~\ref{asmp::LP1}--\ref{asmp::LP3}, for all $n\geq n_0$, $h\geq 1/(2n)$, and $t \in [0,1]$, the weights $W_{ni}$ in \eqref{eqn::W_def} satisfy:
\begin{itemize}
\item[i.] $\sup_{i,t} |W_{ni}(t;h)|\leq C_3/nh$;
\item[ii.] $\sum_{i=1}^n |W_{ni}(t;h)| \leq C_3$,
\end{itemize} 
where the constant $C_3$ depends only on $\lambda_0$ and $K_{\max}$.
\end{lemma}

Recall that we also assume the unknown true solutions $X_j^*$, $j=1,\ldots, p$, belong to a H\"{o}lder class in Assumption~\ref{asmp::smoothness} in the main paper. We state the definition here for completeness.
\begin{definition}\label{def::holder}
Let $T$ be an interval in $\mathds{R}$ and let $\beta_1$ and $L_1$ be two positive numbers. The H\"{o}lder class $\Sigma(\beta_1, L_1)$ on $T$ is defined as the set of $\ell=\lfloor \beta_1 \rfloor$ times differentiable functions $f: T \rightarrow \mathds{R}$ whose $\ell$th order derivative $f^{(\ell)}(\cdot)$ satisfies 
\begin{equation*}
| f^{(\ell)}(x)-f^{(\ell)}(x')| \leq L_1 | x- x'|^{\beta_1-\ell}, \quad \forall x, x' \in T.
\end{equation*}
\end{definition}
We are now ready to prove Theorem~\ref{thm::localP} of the main paper.

\begin{proof}

\begin{align*}
\vertiii{\hat{X}_j- X_j^*}^2 & = \int_0^1 \{\hat{X}_j(u;h)- X_j^*(u) \}^2 \, du = \int_0^1 \left\{ \sum_{i=1}^n Y_{ij}W_{ni}(u;h) - X_j^*(u) \right\}^2 \, du\\
&=  \int_0^1 \left[ \sum_{i=1}^n \{X_{j}^*(t_i) +\epsilon_{ji} \} W_{ni}(u;h) - X_j^*(u) \right]^2 \, du.\\
\end{align*}
Using the property \eqref{eqn::W_sum} of the weights $W_{ni}$ and the fact that $(a+b)^2 \leq 2a^2+2b^2$,
\begin{equation}\label{eqn::decomposition}
\begin{aligned}
\vertiii{\hat{X}_j- X_j^*}^2 & \leq  2 \int_0^1 \left[ \sum_{i=1}^n \{X_{j}^*(t_i)- X_j^*(u) \} W_{ni}(u;h) \right]^2 \, du  \\
& + 2 \int_0^1 \left\{ \sum_{i=1}^n \epsilon_{ji} W_{ni}(u;h)  \right\}^2 \, du\\
& \equiv 2\int_0^1 \text{bias}^2(u) \, du + 2  \int_0^1 g^2( \epsilon_j/\sigma,u,h)  \, du,
\end{aligned}
\end{equation}
where 
\begin{align}
\text{bias}(u) = & \sum_{i=1}^n \{X_{j}^*(t_i)- X_j^*(t) \} W_{ni}(u;h), \label{eqn::biasdef} \\
g( a,u,h ) = &  \sigma\sum_{i=1}^n a_i W_{ni}(u;h), \;\; 
\epsilon_j =  (\epsilon_{1j}, \ldots, \epsilon_{nj})^{\T}, \label{eqn::gdef}
\end{align}
and where $\sigma$ is defined in Assumption~\ref{asmp::errors} in the main paper. 

In what follows, for convenience, we denote the $\ell$th derivative of $X_j^*(t)$ as $X_j^{(\ell)}$. 
Under Assumption~\ref{asmp::smoothness} in the main paper and Assumptions~\ref{asmp::LP1}--\ref{asmp::LP3}, 
it follows from Proposition 1.13 in \citet{tsybakov2009} that $|\text{bias}(u)|\leq q_1 h^{\beta_1}$, where $q_1 = C_3 L_1/{\ell !}$. 
Therefore, 
\begin{equation}\label{eqn::bias}
\int_0^1 \text{bias}^2(u) \, du  \leq q_1^2 h^{2\beta_1}.
\end{equation}

Next, we bound $g(\epsilon_j/\sigma,t,h)$ in \eqref{eqn::gdef} using Theorem 5.6 in \cite{boucheron2013}. The theorem states that if $Z=(Z_1,\ldots, Z_n)$ is a vector of $n$ independent standard normal random variables and $f$ is an $L$-Lipschitz function, then for all $v>0$, 
\begin{equation*}
\text{Pr}\{f(Z)-\mathbb{E} f(Z) \geq v \} \leq \exp\{-v^2/(2L^2) \}.
\end{equation*}
Applying the theorem to $f(z)$ and $-f(z)$, we get
\begin{equation*}
\text{Pr}\{|f(Z)-\mathbb{E} f(Z)| \geq v \} \leq 2\exp\{-v^2/(2L^2) \}.
\end{equation*}

We now show that $ g(x,t,h)$ is an $L_3$-Lipschitz function with $L_3=\sigma C_{3} (nh)^{-0.5}$:
\begin{equation*}
\begin{aligned}
|g(a,u,h)-g(b,u,h)| = & \sigma \left|\sum_{i=1}^{n} (a_i-b_i)W_{ni}(u;h)\right|\\
 \leq & \sigma \left\{ \sum_{i=1}^n W^2_{ni}(u;h) \right\}^{\frac{1}{2}} \|a-b\|_2\\
 \leq & \sigma \left\{\sup_{i,u} |W_{ni}(u;h)| \sum_{i=1}^n |W_{ni}(u,h)| \right\}^{\frac{1}{2}} \|a-b\|_2\\
 \leq & \sigma C_{3} \sqrt{\frac{1}{nh}} \|a-b\|_2,
\end{aligned}
\end{equation*}
where the last inequality follows from Lemma~\ref{lmm::W}. Hence, from Theorem 5.6 in \cite{boucheron2013}, we have 
\begin{equation*}
\text{Pr}\{|{g}(\epsilon_j/\sigma,u,h)-\mathbb{E}{g}(\epsilon_j/\sigma,u,h)|  \geq v \} \leq 2\exp\{-v^2/(2L_3^2)\}.
\end{equation*}
Letting $v=n^{\alpha/2-0.5}h^{-0.5}$  and noting that $\mathbb{E}[g(\epsilon_j/\sigma,u,h)]=0$, we have 
\begin{equation}\label{eqn::g_concen}
\text{Pr}\{|{g}(\epsilon_j/\sigma,u,h)|  \geq n^{\alpha/2-0.5}h^{-0.5} \} \leq  2\exp\{-n^{\alpha}/(2\sigma^2 C^2_{3})\}.
\end{equation}

Combining \eqref{eqn::decomposition}, \eqref{eqn::bias}, and   \eqref{eqn::g_concen}, we have 
\begin{equation}\label{eqn::bound}
\begin{aligned}
\vertiii{\hat{X}_j-X_j^*}^2 & \leq 2\int_0^1 \text{bias}^2(u) \, du +2  \int_0^1 g^2( \epsilon_j/\sigma,u,h)  \, du \\
& \leq 2 q_1^2 h^{2\beta_1} + 2 n^{\alpha-1}h^{-1},
\end{aligned}
\end{equation}
with probability at least $1-2\exp\{-n^{\alpha}/(2\sigma^2 C^2_{3})\}$.

Minimizing the right-hand side of \eqref{eqn::bound} with respect to $h$, we find that the minimizer $h_n$ satisfies   
\begin{equation*}
2\beta_1 q_1^2 {h}_n^{2\beta_1+1} = n^{\alpha-1} .
\end{equation*}
Thus, for ${h}_n \propto n^{(\alpha-1)/(2\beta_1+1)}$, the error bound is 
\begin{equation*}
\vertiii{\hat{X}_j-X_j^*}^2 \leq  C_2 n^{\frac{2\beta_1 }{2\beta_1+1}(\alpha-1)}, 
\end{equation*}
for some global constant $C_2$.
\end{proof}

\subsection{Variable selection consistency of group lasso in error-in-variable models}\label{sec::group}

We first review some notation that is heavily used in this section.  In \eqref{eqn::Psihat} of the main paper, we made use of the notation 
$$ \hat{\Psi}_0(t)=t; \ \hat{\Psi}_k(t) = \int_0^t \psi(\hat{X}_k(u;h)) \, du, \ k=1, \ldots, p.$$
Therefore, $\hat{\Psi}_k(t)$ is an $M$-vector  for $k=1,\ldots, p$ and a scalar for $k=0$. We sometimes use sets, e.g. $S_j$ and $S_j^0$, as the subscripts. In this case, $\hat{\Psi}_{S_j}(t)$ is an $Ms_j$-vector, which is composed of $\hat{\Psi}_{k}$ for $k\in S_j$. Furthermore, $\hat{\Psi}_{S_j^0} = (\hat{\Psi}_0(t), \hat{\Psi}_{S_j}^{\T}(t))^{\T}$ is an $(Ms_j+1)$-vector. Without subscripts, $\hat{\Psi}(t)\equiv (\hat{\Psi}_0(t),\hat{\Psi}_1^{\T}(t), \ldots, \hat{\Psi}_p^{\T}(t))^{\T}$ is of dimension $Mp+1$. We will also apply subscripts to the quantities  $\theta_j^*$, $\hat{\theta}_j$, $\hat{g}$, and $R$. For instance, $\hat{\theta}_{jk}=(\theta_{jk1},\ldots, \theta_{jkM})^{\T}$ for $k = 1, \ldots, p$, and $\hat{\theta}_j=(\hat{\theta}_{j0}, \hat{\theta}_{j1}^{\T}, \ldots, \hat{\theta}_{jp}^{\T})^{\T}$. The products of these vectors are defined as usual, e.g.,  $\hat{\theta}_{j S^0_j}^{\T} \hat{\Psi}_{S_j^0}(t)$ is a scalar, and $\hat{\Psi}_{S_j^0}(t) \hat{\Psi}_{S_j^0}^{\T}(t) $ is an $(Ms_j+1) \times (Ms_j+1)$ matrix.

The optimization problem \eqref{eqn::us_objective} in the main paper is a standardized group lasso problem \citep{simon2012}. Because the regressors $\hat{\Psi}_1,\ldots, \hat{\Psi}_p$ are estimated, establishing variable selection consistency requires extra attention. For ease of discussion, we re-state the optimization problem \eqref{eqn::us_objective},
\begin{equation*} \small
\begin{aligned}
\hat{\theta}_j= \underset{C_0\in \mathbb{R},\theta_{j0} \in \mathbb{R}, \  \theta_{jk} \in \mathbb{R}^{M}}{\arg \min} \ &  \frac{1}{2n} \sum_{i=1}^n \left\{Y_{ij} - C_0- \theta_{j0}\hat{\Psi}_0(t_i) -\sum_{k=1}^p  \theta_{jk}^{\T} \hat{\Psi}_{k}(t_i)\right\}^2+ \\
& \lambda_{n,j} \sum_{k=1}^p \left[ \frac{1}{n} \sum_{i=1}^n  \{\theta_{jk}^{\T}\hat{\Psi}_{k}(t_i)\}^2 \right]^{1/2},  
\end{aligned}
\end{equation*}
where
\begin{equation*}
 \hat{X}(\cdot;h) = \underset{{Z}(\cdot) \in \mathcal{X}(h)}{\arg \min} \sum_{i=1}^n \|Y_i- Z(t_i)\|_2^2, \end{equation*}
\begin{equation*}
\hat{\Psi}_0(t)=t; \ \hat{\Psi}_k(t) = \int_0^t \psi(\hat{X}_k(u;h)) \, du, \ k=1, \ldots, p. 
\end{equation*}
In what follows, for simplicity we assume that $X^*_j(0)=0$, and that $\lambda_{n,1} = \cdots = \lambda_{n,p} \equiv \lambda_n$. 
For any $1\leq j,k \leq p$, let $\theta_{jk}^* \in \mathbb{R}^M$ be the coefficients of the true functions $f_{jk}^*$ on the bases $\psi(\cdot)$, i.e.,
\begin{equation}\label{eqn::thetastar}
f_{jk}^*(a) = \psi(a)^{\T} \theta_{jk}^*+\delta_{jk}(a),
\end{equation}
where $f_{jk}^*$ is introduced in Assumption~\ref{asmp::additivity} in the main paper.
Here we establish variable selection consistency for group lasso regression with errors in variables.  We extend the recent work of  \cite{loh2012} for lasso regression; related results can be found in \cite{ma2010} and \cite{rosenbaum2010}.  In order for variable selection consistency to hold, we need four conditions. In Section~\ref{sec::proof_main}, we will show that these conditions hold with high probability given Assumptions~\ref{asmp::errors}--\ref{asmp::thetamin} in the main paper and Assumptions~\ref{asmp::LP1}--\ref{asmp::psibound}. 

\begin{condition}\label{cond::coherence}
Suppose that 
\begin{equation*}
\begin{aligned}
& 0< \frac{1}{2}C_{\min} \leq \Lambda_{\min} \left( \frac{1}{n} \sum_{i=1}^n \hat{\Psi}_{S_j^0}(t_i) \hat{\Psi}_{S_j^0}^{\T}(t_i) \right), \\
& \Lambda_{\max} \left(\frac{1}{n} \sum_{i=1}^n \hat{\Psi}_{S_j^0}(t_i) \hat{\Psi}_{S_j^0}^{\T}(t_i)\right) \leq 2C_{\max},\\
& 0< \frac{1}{2}C_{\min} \leq \Lambda_{\min} \left( \frac{1}{n} \sum_{i=1}^n \hat{\Psi}_{k}(t_i) \hat{\Psi}_{k}^{\T}(t_i) \right), \quad k \notin S_j^0,
\end{aligned}
\end{equation*}
where $C_{\min}$ and $C_{\max}$ are introduced in Assumption~\ref{asmp::coherence_pop} in the main paper. 
\end{condition}

\begin{condition}\label{cond::irrepresentability}
Assume that 
\begin{equation*}
\max_{k \notin S^0_j} \left\|  \left(\frac{1}{n} \sum_{i=1}^n  \hat{\Psi}_{k}(t_i) \hat{\Psi}_{S_j^0}^{\T}(t_i) \right) \left( \frac{1}{n} \sum_{i=1}^n  \hat{\Psi}_{S_j^0}(t_i) \hat{\Psi}_{S_j^0}^{\T}(t_i) \right)^{-1}   \right\|_2 \leq  2\xi,
\end{equation*}
where $\xi$ is introduced in Assumption~\ref{asmp::irrepresentability_pop}. 
\end{condition}

The next condition  was first proposed in \cite{loh2012} as the deviation condition. Specifically, \eqref{eqn::deviation} is a special case of Equation 3.1 in \cite{loh2012}. 
Recall that the true parameters $\theta_{j0}^*$ and $\theta_{jk}^*$ are introduced in Assumption~\ref{asmp::additivity} of the main paper and \eqref{eqn::thetastar}, respectively.    
\begin{condition}\label{cond::deviation_cond} 
For $j = 1,\ldots, p$, let $\Delta \equiv \max_{j=1,\ldots,p} \vertiii{\hat{X}_j-X_j^*}$. Assume that 
\begin{equation}\label{eqn::deviation}
\left\| \frac{1}{n} \sum_{i=1}^n \hat{\Psi}_k(t_i) Y_{ij}  - \frac{1}{n} \sum_{i=1}^n \hat{\Psi}_k(t_i) \hat{\Psi}_{S^0_j}^{\T}(t_i)\theta^*_{jS^0_j}   \right\|_2 \leq \eta, \ k=0, \ldots, p
\end{equation}
where  $\eta = M^{1/2} \left\{s M^{-\beta_2} Q^{1/2} B+BD\|\theta^*_{S}\|_1 \Delta  +n^{{\alpha}/{2}-1/2 }  \right\}$.
\end{condition}
Note that the global constant $Q$ in Condition~\ref{cond::deviation_cond} also appears in Assumption~\ref{asmp::psibound} in Section~\ref{sec::lemmas}.

Condition~\ref{cond::set_group} places constraints on the quantities involved in the proof of Theorem~\ref{thm::group}. In the proof of Theorem~\ref{thm::main} in the main paper, we will show that Condition~\ref{cond::set_group} holds with high probability. 
\begin{condition}\label{cond::set_group}
The following inequalities hold:
\begin{equation*}
\frac{2\sqrt{s+1}}{C_{\min}} \eta + \lambda_n\frac{\sqrt{8sC_{\max}}}{C_{\min}} \leq \frac{2}{3}\theta_{\min},
\end{equation*}
\begin{equation*}
\frac{2\xi\sqrt{s+1}+1}{\lambda_n}\eta +2\xi \sqrt{s} \sqrt{2C_{\max}} <\sqrt{C_{\min}/2},
\end{equation*}
where $\theta_{\min}\equiv \min_{k \in S_j^0} \|\theta_{jk}^*\|_2$, and $\xi, \eta, C_{\min}$, and $C_{\max}$ are introduced in Assumptions~\ref{asmp::coherence_pop}--\ref{asmp::thetamin} of the main paper.
\end{condition}

We  arrive at the following theorem. 
\begin{theorem}\label{thm::group}
Suppose that Conditions~\ref{cond::coherence}--\ref{cond::set_group} are met. Then the estimator $\hat{\theta}_{j}$ from \eqref{eqn::us_objective} has the correct support, i.e. $\hat{S}_j=S_j$\  for all $j=1,\ldots, p$.
\end{theorem}

\begin{proof}
 We establish variable selection consistency  using the primal-dual witness method \citep{wainwright2009}. For simplicity, we drop the subscript $j$ in what follows: for instance, we drop the subscript $j$ in $Y_{ij}$ and $\hat{\theta}_{j}$ in \eqref{eqn::us_objective}, and in the estimated neighbourhood $\hat{S}_j$. 

A vector $\hat{\theta}$ solves the optimization problem \eqref{eqn::us_objective} in the main paper if it satisfies the Karush-Kuhn-Tucker (KKT) condition, which is 
\begin{equation}\label{eqn::subgrad}
\frac{1}{n} \sum_{i=1}^n \hat{\Psi}_k(t_i)  \left\{ \hat{\Psi}^{\T}(t_i) \hat{\theta} - Y_{i}\right \} + \lambda_n \hat{g}_k =0, \quad k=1,\ldots, p,
\end{equation}
with 
\begin{equation}\label{eqn::g}
\begin{aligned}
\hat{g}_k=\frac{ \sum_{i=1}^n \hat{\Psi}_k(t_i) \hat{\Psi}_k^{\T}(t_i) \hat{\theta}_k/n }{ \sqrt{   \hat{\theta}^{\T}_k \sum_{i=1}^n \hat{\Psi}_k(t_i) \hat{\Psi}_k^{\T}(t_i)                                                                 \hat{\theta}_k/n }} \quad & \text{if} \ \hat{\theta}_k \neq 0, \\
\hat{g}^{\T}_k \left(\frac{1}{n}\sum_{i=1}^n \hat{\Psi}_k(t_i) \hat{\Psi}_k^{\T}(t_i) \right)^{-1} \hat{g}_k < 1 \quad & \text{if} \  \hat{\theta}_k =0.
\end{aligned}
\end{equation}
The KKT condition for $\hat{\theta}_0$ is                                                                                      
\begin{equation}\label{eqn::subgrad_0}
\frac{1}{n} \sum_{i=1}^n \hat{\Psi}_0(t_i)  \left\{ \hat{\Psi}^{\T}(t_i) \hat{\theta}- Y_{i}\right \} =0. 
\end{equation} 
Note that, in the previous equations, we drop the parameter $C_0$ that appears in \eqref{eqn::us_objective} of the main paper to avoid cumbersome bookkeeping. 

We will construct an oracle estimator $\hat{\theta}$ and will verify that it satisfies the KKT conditions \eqref{eqn::subgrad}, \eqref{eqn::g}, and \eqref{eqn::subgrad_0}, which means that it solves the optimization problem \eqref{eqn::us_objective} in the main paper. 

We construct an oracle primal-dual pair $(\hat{\theta},\hat{g})$ as follows: 
\begin{enumerate}
\item Set $\hat{\theta}_{k}=0$ for $k \notin S^0$.
\item Let 
\begin{equation}\label{eqn::constrained}
\begin{aligned}
\hat{\theta}_{S^0} =&  \underset{\theta_{S^0} \in \mathbb{R}^{s M +1}}{\arg \min}  \frac{1}{2n} \sum_{i=1}^n  \left\{Y_{i} - \theta_{S^0}^{\T} \hat{\Psi}_{S^0}(t_i)   \right\}^2 + \lambda_n \sum_{k \in S} \left[ \frac{1}{n} \sum_{i=1}^n  \{\theta_{jk}^{\T}\hat{\Psi}_{k}(t_i)\}^2 \right]^{1/2}. 
\end{aligned} 
\end{equation}
\item Define $\hat{g}_{S^0}=(0, \hat{g}_S^{\T})^{\T}$ as in \eqref{eqn::g}.
\item Solve $\hat{g}_{k}$ from the sub-gradient condition \eqref{eqn::subgrad} for $k \notin S^0$.
\end{enumerate}

We will verify the support recovery consistency 
\begin{equation}\label{eqn::support_rec}
\max_{k \in S} \| \hat{\theta}_k -\theta_k^*\|_2 \leq \frac{2}{3}\theta_{\min}
\end{equation}
 and strict dual feasibility 
\begin{equation}\label{eqn::strict_feas}
\max_{k \notin S^0} \hat{g}^{\T}_k \left( \frac{1}{n}\sum_{i=1}^n \hat{\Psi}_k(t_i) \hat{\Psi}_k^{\T}(t_i) \right)^{-1} \hat{g}_k < 1.
\end{equation}
\eqref{eqn::support_rec} implies that the oracle estimator $\hat{\theta}$ recovers the support of $\theta^*$ exactly, and \eqref{eqn::strict_feas} implies that $\hat{\theta}$ solves \eqref{eqn::us_objective}. 

Further, if the optimal solution to \eqref{eqn::us_objective} is unique, then the oracle estimator is the unique estimator. If the optimal solution is not unique, then from Theorem 2 in \cite{roth2008}, the null set of any optimal solution should contain $S^c$, and thus any optimal solution satisfies the construction of the oracle estimator. Therefore, the statement of Theorem~\ref{thm::group} holds for any optimal solution for \eqref{eqn::us_objective}.

We now establish \eqref{eqn::support_rec}. The subgradient condition for the constrained problem \eqref{eqn::constrained} is 
\begin{equation}\label{eqn::KKT_constrained}
\frac{1}{n}\sum_{i=1}^n \hat{\Psi}_{S^0}(t_i) \big\{\hat{\Psi}_{S^0}^{\T}(t_i)\hat{\theta}_{S^0}  -  Y_{i} \big\} 
+\lambda_n \hat{g}_{S^0} = 0.
\end{equation} 
Adding and subtracting $\frac{1}{n}\sum_{i=1}^n \hat{\Psi}_{S^0}(t_i) \hat{\Psi}_{S^0}^{\T}(t_i)\theta_{S^0}^*$, we get 
\begin{equation*}
\begin{aligned}
& \frac{1}{n}\sum_{i=1}^n\left\{ \hat{\Psi}_{S^0}(t_i) \hat{\Psi}_{S^0}^{\T}(t_i) \hat{\theta}_{S^0} - \hat{\Psi}_{S^0}(t_i) \hat{\Psi}_{S^0}^{\T}(t_i) \theta_{S^0}^* \right\} + \\
&\frac{1}{n}\sum_{i=1}^n\left\{ \hat{\Psi}_{S^0}(t_i) \hat{\Psi}_{S^0}^{\T}(t_i)\theta_{S^0}^*- \hat{\Psi}_{S^0}(t_i) Y_{i}\right\} +\lambda_n \hat{g}_{S^0} = 0.
\end{aligned}
\end{equation*}
Rearranging the terms and letting  
\begin{equation}\label{eqn::RS0}
R_{S^0} \equiv \frac{1}{n}\sum_{i=1}^n \hat{\Psi}_{S^0}(t_i) \hat{\Psi}_{S^0}^{\T}(t_i) \theta_{S^0}^*-\frac{1}{n}\sum_{i=1}^n \hat{\Psi}_{S^0}(t_i) Y_{i}, 
\end{equation}
we get
\begin{equation}\label{eqn::betaS}
\hat{\theta}_{S^0} - \theta^*_{S^0}= -\left( \frac{1}{n}\sum_{i=1}^n \hat{\Psi}_{S^0}(t_i) \hat{\Psi}_{S^0}^{\T}(t_i)\right)^{-1} \left( R_{S^0}  + \lambda_n \hat{g}_{S^0} \right). 
\end{equation}

By the definition of $R_{S^0}$ in \eqref{eqn::RS0}, for each $k \in S$, we have that
\begin{equation}\label{eqn::Rk_defn}
R_k = \frac{1}{n}\sum_{i=1}^n \hat{\Psi}_k(t_i) \hat{\Psi}_{S^0}^{\T}(t_i) \theta_{S^0}^*-\frac{1}{n}\sum_{i=1}^n \hat{\Psi}_k(t_i) Y_{i},
\end{equation}
and $R_{0}=\frac{1}{n}\sum_{i=1}^n t_i \{\hat{\Psi}_{S^0}^{\T}(t_i){\theta}_{S^0}^*  -  Y_{i} \}$. By Condition~\ref{cond::deviation_cond}, we know that $\|R_k\|_2 \leq \eta$ for $k \in S^0$.  Hence, 
\begin{equation}\label{eqn::RS}
\|R_{S^0}\|_2 \leq \eta \sqrt{s+1}. 
\end{equation}
By Condition~\ref{cond::coherence}, we have that 
\begin{equation}\label{eqn::maxinverse}
\Lambda_{\max}\left\{  \left( \frac{1}{n}\sum_{i=1}^n \hat{\Psi}_{S^0}(t_i) \hat{\Psi}_{S^0}^{\T}(t_i) \right)^{-1} \right\} \leq \frac{2}{C_{\min}}.
\end{equation}
From \eqref{eqn::g} and the fact that the largest eigenvalue of a submatrix is no greater than the largest eigenvalue of the matrix,
\begin{equation*}
\frac{1}{2C_{\max}} \|\hat{g}_k\|_2^2   \leq \hat{g}^{\T}_k \left( \frac{1}{n}\sum_{i=1}^n \hat{\Psi}_k(t_i) \hat{\Psi}_k^{\T}(t_i) \right)^{-1} \hat{g}_k  = 1, \quad k \in S.
\end{equation*}
Furthermore, $\hat{g}_0=0$ by construction. Hence,  
\begin{equation}\label{eqn::g_bound}
 \|\hat{g}_{S^0}\|_2 = \left\{ \|\hat{g}_0\|_2^2+ \|\hat{g}_S\|_2^2 \right\}^{1/2} \leq \sqrt{2 s C_{\max}}.
\end{equation}
Therefore,  combining  \eqref{eqn::betaS}, \eqref{eqn::RS}, \eqref{eqn::maxinverse}, and \eqref{eqn::g_bound}, it follows that 
\begin{equation*}
\max_{k \in S} \|\hat{\theta}_{k} -\theta_{k}^*\|_2 \leq \|\hat{\theta}_{S^0} -\theta_{S^0}^*\|_2 \leq \frac{2\eta\sqrt{s+1}}{C_{\min}}  + \lambda_n\frac{\sqrt{8s C_{\max}}}{C_{\min}}\leq \frac{2}{3}\theta_{\min},
\end{equation*}
where the last inequality follows from Condition~\ref{cond::set_group}.

Next, we verify  strict feasibility \eqref{eqn::strict_feas}. For $k \notin S^0$, from \eqref{eqn::subgrad},
\begin{equation*}
\frac{1}{n}\sum_{i=1}^n \hat{\Psi}_k(t_i) \big( \hat{\Psi}_{S^0}^{\T}(t_i) \hat{\theta}_{S^0} - Y_{i} \big) +\lambda_n \hat{g}_k = 0.
\end{equation*}
Adding and subtracting $\frac{1}{n}\sum_{i=1}^n \hat{\Psi}_{k}(t_i) \hat{\Psi}_{S^0}^{\T}(t_i)\theta_{S^0}^*$ yields
\begin{equation*}
\begin{aligned}
&\frac{1}{n}\sum_{i=1}^n \left\{ \hat{\Psi}_k(t_i) \hat{\Psi}_{S^0}^{\T}(t_i) \hat{\theta}_{S^0} -  \hat{\Psi}_k(t_i) \hat{\Psi}_{S^0}^{\T}(t_i) \theta_{S^0}^*\right\} +\\
&\frac{1}{n}\sum_{i=1}^n \left\{ \hat{\Psi}_k(t_i) \hat{\Psi}_{S^0}^{\T}(t_i) \theta_{S^0}^*- \hat{\Psi}_k(t_i)Y_{i}\right\} +\lambda_n \hat{g}_k = 0.
\end{aligned}
\end{equation*}
Rearranging the terms and plugging in \eqref{eqn::betaS} and \eqref{eqn::Rk_defn}, we get
\begin{equation*}
\lambda_n \hat{g}_k = \frac{1}{n}\sum_{i=1}^n \hat{\Psi}_k(t_i) \hat{\Psi}_{S^0}^{\T}(t_i) \left( \frac{1}{n}\sum_{i=1}^n \hat{\Psi}_{S^0}(t_i) \hat{\Psi}_{S^0}^{\T}(t_i) \right)^{-1} (R_{S^0}+ \lambda_n \hat{g}_{S^0}) - R_k.
\end{equation*}

By Condition~\ref{cond::irrepresentability}, we know that 
\begin{equation*}
\max_{k \notin S^{0}} \left\|  \left( \frac{1}{n}\sum_{i=1}^n \hat{\Psi}_k(t_i) \hat{\Psi}_{S^0}^{\T}(t_i) \right) \left( \frac{1}{n}\sum_{i=1}^n \hat{\Psi}_{S^0}(t_i) \hat{\Psi}_{S^0}^{\T}(t_i) \right)^{-1}   \right\|_2 \leq  2\xi.
\end{equation*}
Recall from Condition~\ref{cond::deviation_cond} that $\|R_k\|_2 \leq \eta$ for $ 1\leq k \leq p$. 
Using \eqref{eqn::RS} and \eqref{eqn::g_bound}, we have that 
\begin{equation*}
\|\hat{g}_k\|_2 \leq \frac{2\xi\sqrt{s+1}+1}{\lambda_n}\eta +2\xi \sqrt{s} \sqrt{2C_{\max}}, \quad k \notin S^0.
\end{equation*}
By Condition~\ref{cond::set_group}, $\|\hat{g}_k\|_2< \sqrt{C_{\min}/2}$, and thus, applying Condition~\ref{cond::coherence},
\begin{equation*}
\hat{g}^{\T}_k \left( \frac{1}{n}\sum_{i=1}^n \hat{\Psi}_k(t_i) \hat{\Psi}_k^{\T}(t_i) \right)^{-1} \hat{g}_k  \leq  \frac{2\|\hat{g}_k\|_2^2}{C_{\min}} <1, \quad k \notin S^0.
\end{equation*}
Therefore, we have established \eqref{eqn::strict_feas}.
\end{proof}

\subsection{Assumption~\ref{asmp::psibound} and technical lemmas}\label{sec::lemmas}

Theorem~\ref{thm::group} characterizes the samples on which the GRADE estimator is able to reconstruct the true network. We must now establish that with high probability, the observations satisfy Conditions~\ref{cond::coherence}--\ref{cond::set_group}. In Section~\ref{sec::proof_main}, Lemmas~\ref{lmm::coherence_sample}--\ref{lmm::deviation}, stated below, will be used to show that Conditions~\ref{cond::coherence}--\ref{cond::set_group}, needed for Theorem~\ref{thm::group}, hold with high probability. Lemma~\ref{lmm::basis} is used to prove Lemmas~\ref{lmm::coherence_sample}--\ref{lmm::deviation}. 
Lemmas~\ref{lmm::basis} -- \ref{lmm::deviation} are proven in Appendix~\ref{sec::lemmas_proof}.

First, we state the regularity condition on the bases $\psi$ mentioned in Section~\ref{sec::theory} in the main paper. 
\begin{assumption}\label{asmp::psibound}
The basis functions are orthonormal, i.e., $\int_0^1 \psi_{jk}\big(X_k^*(u)\big) \psi_{jk}^{\T}\big(X_k^*(u)\big) du = I_M$, where $I_M$ is an $M\times M$ identity matrix. 
The basis functions are bounded and have bounded first order derivative, i.e. $|\psi_m(x)| \leq B, \ |\psi_{m}'(x)| \leq D, m=1, \ldots, M$. 
Further, under Assumption~\ref{asmp::additivity} in the main paper, for any $j,k$, 
\begin{equation}\label{eqn::trigono}
\int_0^1 \delta_{jk}^2(u) du = \int_{0}^1 \left\{f^*_{jk}(X_k^*(u))- \psi^{\T}(X_k^*(u)) \theta^*_{jk} \right\}^2 \, du \leq Q (M+1)^{-2\beta_2},
\end{equation}
 where $\theta_{jk}^*$ is defined in \eqref{eqn::thetastar} and $Q$ is a global constant. 
\end{assumption}

\begin{remark}
Assumption~\ref{asmp::psibound} holds, for instance, when $\psi(\cdot)$ is the set of trigonometric basis functions (see, e.g., Section~1.7.3 in \cite{tsybakov2009}).
\end{remark}

We next state the technical lemmas used in the proof of Theorem~\ref{thm::main} in the main paper. 

\begin{lemma}\label{lmm::basis}
Suppose that Assumption~\ref{asmp::additivity} in the main paper and Assumption~\ref{asmp::psibound} hold, and $\psi(t)=(\psi_0(t), \psi_1(t),\ldots, \psi_{M}(t))^{\T}$ is of degree $M$. Then,
\begin{equation}\label{eqn::f2theta}
\left| \| \theta^*_{jk}\|_2 -  \left\{ \int_0^1  \big[f_{jk}^*(X_k^*(u))\big]^2 du \right\}^{1/2} \right| \leq \sqrt{Q} M^{-\beta_2}. 
\end{equation}
\begin{equation}\label{eqn::truncation_error}
\vertiii{X_j^*-  \Psi_{S^0}^{\T}\theta_{S^0}^*} \leq s\sqrt{QM^{-2\beta_2}},
\end{equation}
and 
\begin{equation}\label{eqn::truncation_error_dis}
\frac{1}{n}\sum_{i=1}^n \{X_{j}^*(t_i)-\Psi_{S^0}^{\T}(t_i) \theta_{S^0}^* \}^2\leq  s^2 Q M^{-2\beta_2} + o\left(n^{-2}\right),
\end{equation}
where $\theta_{jk}^*$ is defined in \eqref{eqn::thetastar} and $Q$ is a constant in Assumption~\ref{asmp::psibound}.
\end{lemma}

\begin{lemma}\label{lmm::coherence_sample} 
Suppose that  Assumptions~\ref{asmp::additivity} and \ref{asmp::coherence_pop} in the main paper and  Assumption~\ref{asmp::psibound} hold. Let $ \Delta  \equiv \max_{{j =1, \ldots, p}}\vertiii{\hat{X}_j - X_j^*}$. The following bounds on the eigenvalues of $\sum_{i=1}^n \hat{\Psi}_{S^0} \hat{\Psi}_{S^0}^{\T} /n$ hold:
\begin{equation}\label{eqn::coherence}
\begin{aligned}
\Lambda_{\min}\left( \frac{1}{n} \sum_{i=1}^n \hat{\Psi}_{S^0}(t_i) \hat{\Psi}_{S^0}^{\T}(t_i)  \right) & \geq C_{\min} - \left( 2BD\Delta+ \frac{BD+B^2}{6 n^2} \right)(Ms+1),\  \ \\
 \Lambda_{\max}\left( \frac{1}{n} \sum_{i=1}^n \hat{\Psi}_{S^0}(t_i) \hat{\Psi}_{S^0}^{\T}(t_i) \right) & \leq C_{\max} + \left( 2BD\Delta+ \frac{BD+B^2}{6 n^2} \right)(Ms+1), \\
\text{and}  \quad \Lambda_{\min}\left( \frac{1}{n} \sum_{i=1}^n \hat{\Psi}_{k}(t_i) \hat{\Psi}_{k}^{\T}(t_i)  \right) & \geq C_{\min} - \left( 2BD\Delta+ \frac{BD+B^2}{6 n^2} \right)M, \quad k \notin S_j^0.
\end{aligned}
\end{equation}
\end{lemma}

\begin{lemma}\label{lmm::irrepresentability_sample}
Suppose that  Assumptions~ \ref{asmp::additivity} and \ref{asmp::irrepresentability_pop} in the main paper and  Assumption~\ref{asmp::psibound}  hold. Let $ \Delta  \equiv \max_{{j =1, \ldots, p}}\vertiii{\hat{X}_j - X_j^*}$. Then, 
\begin{equation}\label{eqn::irrepresentability}
\begin{aligned}
& \left\| \left( \frac{1}{n} \sum_{i=1}^n\hat{\Psi}_k(t_i) \hat{\Psi}_{S^0}^{\T}(t_i) \right) \left(\frac{1}{n} \sum_{i=1}^n\hat{\Psi}_{S^0}(t_i)\hat{\Psi}_{S^0}^{\T}(t_i)  \right)^{-1} \right\|_2  \leq \\
 & \xi+\left\{c_1 \hat{C}_{\min}^{-2} M(Ms+1)^3 \Delta^2\right\}^{1/2}+\left\{c_2 M(Ms+1)\Delta^2\right\}^{1/2}+ \left\{ c_3 M(Ms+1)^3/{6n^2} \right\}^{1/2},
 \end{aligned}
\end{equation}
where $\hat{C}_{\min}\equiv  C_{\min} - \left( 2BD\Delta+ \frac{BD+B^2}{6 n^2} \right)(Ms+1)$, and $c_1, c_2, c_3$ are constants.
\end{lemma}

\begin{lemma}\label{lmm::deviation}
Suppose Assumptions~\ref{asmp::errors},  \ref{asmp::smoothness}, and \ref{asmp::additivity} in the main paper and Assumption~\ref{asmp::psibound} hold. Let $\Delta \equiv \max_{j=1,\ldots,p} \vertiii{\hat{X}_j-X_j^*}$. For each $k=0,\ldots,p$, 
\begin{equation}\label{eqn::eta}
\left\| \frac{1}{n} \sum_{i=1}^n \hat{\Psi}_k(t_i) Y_{ij}  - \frac{1}{n} \sum_{i=1}^n \hat{\Psi}_k(t_i) \hat{\Psi}_{S^0}^{\T}(t_i)\theta^*_{S^0}   \right\|_2 \leq \eta,
\end{equation} 
where $$\eta \equiv M^{1/2} \left\{s M^{-\beta_2} Q^{1/2} B+BD\|\theta^*_{S}\|_1 \Delta  +n^{{\alpha}/{2}-{1}/{2} }  \right\}$$
 with probability at least $1-2M\exp\{-n^{\alpha}/(2B^2\sigma^2) \}$.
\end{lemma}

\subsection{Proof of Theorem~\ref{thm::main}}\label{sec::proof_main}

\begin{proof}
Notice that Theorem~\ref{thm::group} offers the desired result of Theorem~\ref{thm::main} in the main paper. We now verify that Conditions~\ref{cond::coherence}--\ref{cond::set_group} hold with high probability given the assumptions for Theorem~\ref{thm::main} of the main paper. This completes the  proof of  Theorem~\ref{thm::main} of the main paper. 

First of all, Lemma~\ref{lmm::deviation} tells us that Condition~\ref{cond::deviation_cond} holds with probability at least $1-2pM\exp^{-n^{\alpha}/(2B^2\sigma^2)}$. This probability converges to unity as $p$ and $n$ grow, 
because  $M\propto n^{\frac{2}{2\beta_2+1}\frac{\beta_1}{2\beta_1+1}(1 -\alpha)} = o(n)$ and $pn\exp(-C_4 n^{\alpha}/\sigma^2)=o(1)$ as required in Theorem~\ref{thm::main} of the main paper, where $C_4\equiv \min \{1/(2B^2),1/(2C_3^2)\}$. Thus, Condition~\ref{cond::deviation_cond} holds with high probability.

Next, we verify that Condition~\ref{cond::set_group} holds with high probability. Given Assumptions~\ref{asmp::errors}--\ref{asmp::smoothness} and \ref{asmp::LP1}--\ref{asmp::LP3}, we know from Theorem~\ref{thm::localP} in the main paper that 
\begin{equation}\label{eqn::xdiff}
\underset{j}{\max} \vertiii{\hat{X}_j- X_j^*} \equiv \Delta = O\left(n^{\frac{\beta_1 }{2\beta_1+1}(\alpha-1)}\right), 
\end{equation}
with probability at least $ 1-2p\exp\{-n^{\alpha}/(2C_3 \sigma^2) \}$.  
Recall that in Theorem~\ref{thm::main} of the main paper we require that $s=O(n^{\gamma})$ and $M \propto n^{\frac{2}{2\beta_2+1}\frac{\beta_1}{2\beta_1+1}(1 -\alpha)}$. 
Furthermore, $\|\theta_k^*\|_1 < \sqrt{M}\|\theta_k^*\|_2 $, and $\|\theta_k^*\|_2$ is bounded by a constant due to the fact that $f_{jk}^*$ is bounded and \eqref{eqn::f2theta}.   
Combining these with \eqref{eqn::xdiff}, we know that the three terms of $\eta$ in Condition~\ref{cond::deviation_cond} satisfy 
$$sM^{-\beta_2+1/2}Q^{1/2}B = O\left(  n^{-\frac{2\beta_2-1}{2\beta_2+1}\frac{\beta_1}{2\beta_1+1}(1 -\alpha) +\gamma}\right),$$ 
$$ M^{1/2} BD \|\theta_S^*\|_1 \Delta = O\left(n^{-\frac{2\beta_2-1}{2\beta_2+1}\frac{\beta_1}{2\beta_1+1}(1 -\alpha)+\gamma}\right),$$
and
$$\ M^{1/2}n^{\alpha/2-1/2} = O\left(n^{(\frac{1}{2\beta_2+1}\frac{\beta_1}{2\beta_1+1} - \frac{1}{2})(1 -\alpha) }\right). $$  
These lead to 
\begin{equation}\label{eqn::deviation_asymp}\footnotesize
\left\| \frac{1}{n} \sum_{i=1}^n \hat{\Psi}_k(t_i) Y_{ij}  - \frac{1}{n} \sum_{i=1}^n \hat{\Psi}_k(t_i) \hat{\Psi}_{S^0}^{\T}(t_i)\theta^*_{S^0}   \right\|_2 \leq \eta = O\left(  n^{-\frac{2\beta_2-1}{2\beta_2+1}\frac{\beta_1}{2\beta_1+1}(1 -\alpha) +\gamma}\right)
\end{equation}
with probability at least $1-2pM\exp\{-n^{\alpha}/(2B^2\sigma^2) \}$  for all $k=0,\ldots, p$, from Lemma~\ref{lmm::deviation}. 

In Theorem~\ref{thm::main} of the main paper, we require that $ \lambda_n \propto n^{-\frac{\beta_1}{2\beta_1+1}\frac{2\beta_2-1}{2\beta_2+1}(1 -\alpha) +2\gamma}.$ 
Given \eqref{eqn::deviation_asymp} and $s=O(n^{\gamma})$, we know that $\sqrt{s}\eta = o(\lambda_n)$. Furthermore, define
\begin{equation}\label{eqn::H1}\small
H_1(\beta_1,\beta_2,\alpha) \equiv \min\left\{ \frac{\beta_1}{2\beta_1+1}\frac{2\beta_2-1}{4\beta_2+2}(1-\alpha),  \frac{2}{3}\frac{\beta_1}{2\beta_1+1}\frac{2\beta_2-3}{2\beta_2+1}(1-\alpha) \right\}. 
\end{equation} 
Then, 
$$ -\frac{\beta_1}{2\beta_1+1}\frac{2\beta_2-1}{2\beta_2+1}(1 -\alpha) +2\gamma  \leq -2H_1(\beta_1,\beta_2,\alpha)+2\gamma. $$
Thus,  $\lambda_n=o(1)$ for $\gamma<H_1(\beta_1,\beta_2,\alpha)$. 
Further notice that $M^{-\beta_2} \propto n^{-\frac{2\beta_2}{2\beta_2+1}\frac{\beta_1}{2\beta_1+1}(1 -\alpha)}=o(1)$, which implies that $\theta_{\min} \geq 3f_{\min}/4$ for sufficiently large $n$ from \eqref{eqn::f2theta} in Lemma~\ref{lmm::basis}. 
As a result, the two inequalities in Condition~\ref{cond::set_group} become 
\begin{equation*}
o(\lambda_n) + \lambda_n\frac{\sqrt{8sC_{\max}}}{C_{\min}} \leq \frac{f_{\min}}{2},
\end{equation*}
\begin{equation*}
o(1) +2\xi \sqrt{s} \sqrt{2C_{\max}} <\sqrt{C_{\min}/2},
\end{equation*}
which hold for sufficiently large $n$ under Assumption~\ref{asmp::thetamin} of the main paper. 

Note that the probability that  \eqref{eqn::xdiff} and \eqref{eqn::deviation_asymp} both hold is at least $1-2pM\exp\{-n^{\alpha}/(2B^2\sigma^2) \}-2p\exp\{-n^{\alpha}/(2C_3^2\sigma^2) \} $. 
Letting $C_4= \min\{1/(2B^2),1/(2C_3^2)\}$, we know from Theorem~\ref{thm::main} that  $p n \exp(-C_4 n^{\alpha}/\sigma^2)=o(1)$. 
Combining this with $M\propto n^{\frac{2}{2\beta_2+1}\frac{\beta_1}{2\beta_1+1}(1 -\alpha)} = o(n)$, we know that $1-2pM\exp\{-n^{\alpha}/(2B^2\sigma^2) \}-2p\exp\{-n^{\alpha}/(2C_3^2\sigma^2) \} $ converges to $1$ as $p$, $s$, and $n$ grow. Therefore, Condition~\ref{cond::set_group} holds with high probability.

Finally, we establish that Conditions~\ref{cond::coherence} and \ref{cond::irrepresentability} hold with high probability. Note that the dominant terms not involving $C_{\min}, C_{\max}$ or $\xi$ in the bounds in \eqref{eqn::coherence} in Lemma~\ref{lmm::coherence_sample} and \eqref{eqn::irrepresentability} in Lemma~\ref{lmm::irrepresentability_sample} involve $sM\Delta$ and $s^{3/2}M^2\Delta$, respectively. Given \eqref{eqn::xdiff}, one can check that
\begin{align}\label{eqn::asymp1}
sM\Delta \propto & \  n^{ \frac{\beta_1}{2\beta_1+1}\frac{2\beta_2-1}{2\beta_2+1}(1-\alpha)+\gamma}=o(1), \ \text{and} \\
 s^{3/2}M^2 \Delta \propto & \   n^{\frac{\beta_1}{2\beta_1+1}\frac{2\beta_2-3}{2\beta_2+1}(1-\alpha)+\frac{3}{2}\gamma}=o(1),\label{eqn::asymp2}
\end{align}
where we have used the fact that  $\beta_2\geq 3$ in Assumption~\ref{asmp::additivity} in the main paper as well as the fact that $\gamma < H_1(\beta_1,\beta_2,\alpha)$ from the statement of Theorem~\ref{thm::main} in the main paper. Since \eqref{eqn::xdiff} and \eqref{eqn::deviation_asymp} hold with high probability, combining the inequalities in Lemmas~\ref{lmm::coherence_sample} and \ref{lmm::irrepresentability_sample} with \eqref{eqn::asymp1} and \eqref{eqn::asymp2}, we see that Conditions~\ref{cond::coherence} and \ref{cond::irrepresentability} hold with high probability  given Assumptions~\ref{asmp::additivity}, \ref{asmp::coherence_pop} and \ref{asmp::irrepresentability_pop} in the main paper. 

In summary, we have shown that Conditions~\ref{cond::coherence}--\ref{cond::set_group} hold with high probability. Applying Theorem~\ref{thm::group} establishes that the GRADE estimator $\hat{S}_j$ in \eqref{eqn::us} in the main paper recovers the true support $S^*_j$. 
\end{proof}

\subsection{Proof of Proposition~\ref{prop::main}}\label{sec::proof_prop}

In Proposition~\ref{prop::main}, the choice of bandwidth $h_n$ is different from that in Theorems~\ref{thm::localP} and \ref{thm::main} of the main paper. In order to prove Proposition~\ref{prop::main} of the main paper, we establish the following concentration inequality for $\vertiii{\hat{X}_j- X_j^*}$, where the bandwidth is chosen as specified in Proposition~\ref{prop::main}  of the main paper.
\begin{proposition}\label{prop::localP}
Suppose that Assumptions~\ref{asmp::errors}--\ref{asmp::smoothness} in the main paper and \ref{asmp::LP1}--\ref{asmp::LP3}  hold. Let $\hat{X}_j$ be the local polynomial regression estimator of order $\ell=\lfloor \beta_1 \rfloor$ with bandwidth 
\begin{equation*}
h_n\propto n^{-1/(2\beta_1+1)}.
\end{equation*} 
There exists a constant $C_2 < \infty$ such that for each $j=1, \ldots, p$, 
\begin{equation*}
 \vertiii{\hat{X}_j- X_j^*}^2 \leq C_2 n^{\alpha-\frac{2\beta_1 }{2\beta_1+1}} 
\end{equation*}
holds with probability at least $ 1-2\exp\{-n^{\alpha}/(2\sigma^2 C^2_{3})\}.$
\end{proposition}
The proof of Proposition~\ref{prop::localP} is similar to that for Theorem~\ref{thm::localP} in the main paper by plugging in $
h_n\propto n^{-1/(2\beta_1+1)}$ in \eqref{eqn::bound}.

Given Proposition~\ref{prop::localP}, the proof of Proposition~\ref{prop::main} in the main paper follows from a similar argument as in the proof of Theorem~\ref{thm::main} in the main paper, and is thus omitted here. The constant $H_2(\beta_1,\beta_2,\alpha)$ is defined as
\begin{equation*}
H_2(\beta_1,\beta_2,\alpha) \equiv \min\left\{ \frac{\beta_1}{2\beta_1+1}\frac{2\beta_2-1}{2\beta_2+1}-\alpha,  \frac{1}{3}\frac{\beta_1}{2\beta_1+1}\frac{2\beta_2-3}{2\beta_2+1}-\alpha\right\}.
\end{equation*}

\section{PROOFS OF TECHNICAL LEMMAS}\label{sec::lemmas_proof}
\subsection{Proof of Lemma~\ref{lmm::basis}}\label{sec::basis}

In this section, in the interest of clarity, we  bring back the subscript $j$ in $\theta_{j}^*, \theta_{jk}^*, \theta_{jS^0}^*$ and $f_{jk}^*$. 

\begin{proof}
Recall that in Assumption~\ref{asmp::psibound}, \eqref{eqn::trigono} says that 
\begin{equation*}
\int_0^1 \delta_{jk}^2(u) du = \int_{0}^1 \left\{f^*_{jk}(X_k^*(u))- \psi^{\T}(X_k^*(u)) \theta^*_{jk} \right\}^2 \, du \leq Q (M+1)^{-2\beta_2}.
\end{equation*} 

 It follows from the triangle inequality that
 \begin{equation*}
\left| \left\{ \int_0^1  \big[\psi^{\T}(X_k^*(u)) \theta^*_{jk}\big]^2 du \right\}^{1/2} - \left\{ \int_0^1  \big[f_{jk}^*(X_k^*(u))\big]^2 du \right\}^{1/2} \right| \leq \sqrt{Q} M^{-\beta_2}. 
 \end{equation*}
 The orthogonality of $\psi$ in Assumption~\ref{asmp::psibound} then leads to \eqref{eqn::f2theta}, i.e.,  
 \begin{equation*}
 \left| \| \theta^*_{jk}\|_2 -  \left\{ \int_0^1  \big[f_{jk}^*(X_k^*(u))\big]^2 du \right\}^{1/2} \right| \leq \sqrt{Q} M^{-\beta_2}. 
 \end{equation*}

 From \eqref{eqn::trigono}, we can also see that
\begin{equation*}
\begin{aligned}
\left| \int_0^t \delta_{jk}(u) \, du \right| \leq & \left\{ \int_0^t \delta^2_{jk}(u) \, du\right\}^{1/2} \left\{ \int_0^t 1^2 \, du\right\}^{1/2} \leq  \left\{ \int_0^1 \delta^2_{jk}(u) \, du\right\}^{1/2}\\
 \leq &  \sqrt{Q(M+1)^{-2\beta_2}} \leq \sqrt{QM^{-2\beta_2}},
\end{aligned}
\end{equation*}
where we use the fact that $t \in [0,1]$. 

Recall from \eqref{eqn::ODE_integrated} in the main paper and \eqref{eqn::thetastar} that
$$ X_j^*(t)=\theta_{j0}^* t +\sum_{k=1}^p \Psi_{k}^{\T}(t)\theta_{jk}^*  + \sum_{k=1}^p\int_0^t \delta_{jk}(u) \, du,$$
where we let $X^*_j(0)=0$ for ease of discussion. We know that both $\theta_{jk}^*$ and $\delta_{jk}$ are  zero for $k \notin S$. Thus, the errors that result from  the use of truncated bases are bounded by
\begin{equation*}
\begin{aligned}
\vertiii{X_j^*-  \Psi_{S_j^0}^{\T}\theta_{jS_j^0}^*} = & \vertiii{X_j^*- \theta_{j0}^* t- \sum_{k \in S_j} \Psi_{k}^{\T}\theta_{jk}^*}=  \left[\int_0^1 \left\{\sum_{k \in S_j} \int_0^t \delta_{jk}(u) du\right\}^2 dt \right]^{\frac{1}{2}} \\
\leq  &\left[\int_0^1 \left\{s \sqrt{QM^{-2\beta_2}} \right\}^2 dt \right]^{\frac{1}{2}} \leq s\sqrt{QM^{-2\beta_2}}.
\end{aligned}
\end{equation*}

The error bound in \eqref{eqn::truncation_error} is on the whole trajectories, whereas we only observe discrete measurements of the trajectories in reality. 
The bound in \eqref{eqn::truncation_error_dis} addresses this case and is proved below.
\begin{equation*}
\begin{aligned}
\frac{1}{n}\sum_{i=1}^n \{X_{j}^*(t_i)-\Psi_{S_j^0}^{\T}(t_i) \theta_{jS_j^0}^* \}^2 &= \frac{1}{n}\sum_{i=1}^n \left\{\sum_{k \in S_j} \int_0^{t_i} \delta_{jk}(u) du\right\}^2\\
& \leq  \int_0^1 \left\{\sum_{k \in S_j} \int_0^t \delta_{jk}(u) du\right\}^2\, dt +  o\left(\frac{1}{n^2} \right) \\
& \leq  s^2 Q M^{-2\beta_2} + o\left(n^{-2} \right),
\end{aligned}
\end{equation*}
where the last inequality follows from \eqref{eqn::truncation_error} and the second to last inequality follows from the trapezoidal rule on a uniform grid. 

\end{proof}

\subsection{Proof of Lemma~\ref{lmm::coherence_sample}}\label{sec::coherence_sample}

We first review some known results on matrix norms and eigenvalues. For an $m\times n$ matrix $A$,  
\begin{equation}\label{eqn::norm}
\|A\|_2 = \sup_{x \in \mathbb{R}^n} \frac{\|Ax\|_2}{\|x\|_2} =\sup_{ \|x\|_2 =1} \left\{ \sum_{i=1}^m \left(\sum_{j=1}^n a_{ij}x_j \right)^2 \right\}^{\frac{1}{2}}
\leq \left(\sum_{i=1}^m \sum_{j=1}^n a_{ij}^2 \right)^{\frac{1}{2}}\equiv\|A\|_F,
\end{equation}
where $\| \cdot\|_F$ is the Frobenius norm. We remind the reader that for a symmetric matrix $A$ that is not positive semi-definite, $\Lambda_{\max}(A)\leq \|A\|_2$. The following two inequalities are useful in the proofs. Let $A$ and $\hat{A}$ be two $n \times n$ symmetric matrices.
\begin{enumerate}
\item Weyl's inequality \citep{weyl1912} states that
\begin{equation*}
 \Lambda_{\min} (A) -\Lambda_{\max}(\hat{A}-A) \leq \Lambda_{\min}(\hat{A}), \ \text{and} \ \Lambda_{\max}(\hat{A}) \leq \Lambda_{\max}(A)+\Lambda_{\max}(\hat{A}-A),
\end{equation*}
which leads to
\begin{equation}\label{eqn::weyl}\small
 \Lambda_{\min} (A) -\|\hat{A}-A\|_2 \leq \Lambda_{\min}(\hat{A}), \ \text{and} \ \Lambda_{\max}(\hat{A}) \leq \Lambda_{\max}(A)+\|\hat{A}-A\|_2.
\end{equation}
\item The Gershgorin circle theorem \citep{gershgorin1931} states that
\begin{equation}\label{eqn::GC}
\|\hat{A}-A\|_2 \leq \max_{i} \sum_{j=1}^{n} |(\hat{A}-A)_{ij}| \leq n\|\hat{A}-A\|_{\infty},
\end{equation}
where the norm $\|\cdot\|_{\infty}$ is defined as $\|A\|_{\infty} = \max_{i,j} |A_{ij}|$.
\end{enumerate}

We are now ready to prove Lemma~\ref{lmm::coherence_sample}.

\begin{proof}
Let $A\equiv \int_0^1 \Psi_{S^0}(t) \Psi_{S^0}^T(t) \, dt$, $A_n \equiv \frac{1}{n}\sum_{i=1}^n \Psi_{S^0}(t_i) \Psi_{S^0}^T(t_i)$, $\hat{A}_n \equiv \frac{1}{n}\sum_{i=1}^n \hat{\Psi}_{S^0}(t_i) \hat{\Psi}_{S^0}^T(t_i)$, which are $(Ms+1)\times (Ms+1)$ matrices. Then, 
\begin{equation}\label{eqn::A_decomp}
\begin{aligned}
\Lambda_{\min} (\hat{A}_n) \geq & \Lambda_{\min}(A)-\|\hat{A}_n-A\|_2 \\
\geq & \Lambda_{\min}(A)-\|A_n-A\|_2 -\|\hat{A}_n - A_n\|_2,
\end{aligned}
\end{equation}
where the first inequality follows from  \eqref{eqn::weyl} and the second follows from the triangle inequality.

Furthermore,
\begin{equation}\label{eqn::Ahatn} \small
\begin{aligned}
\|\hat{A}_n- A_n\|_2 \leq & (M s+1) \| \hat{A}_n- A_n\|_{\infty} \\
\leq & (Ms+1)\left\| \frac{1}{n} \sum_{i=1}^n \left\{ \hat{\Psi}_{S^0}(t_i)\hat{\Psi}_{S^0}^{\T}(t_i) - {\Psi}_{S^0}(t_i){\Psi}_{S^0}^{\T}(t_i) \right\}  \right\|_{\infty}\\
\leq & \frac{Ms+1}{n} \left\| \sum_{i=1}^n \hat{\Psi}_{S^0}(t_i) \left\{ \hat{\Psi}^{\T}_{S^0}(t_i)  -{\Psi}^{\T}_{S^0}(t_i) \right\}\right\|_{\infty} + \\
& \frac{Ms+1}{n} \left\| \sum_{i=1}^n \Psi_{S^0}(t_i) \left\{ \hat{\Psi}^{\T}_{S^0}(t_i)  -{\Psi}^{\T}_{S^0}(t_i) \right\}\right\|_{\infty} \\
\leq & \frac{Ms+1}{n}  \left\| \sum_{i=1}^n \hat{\Psi}_{S^0}(t_i) D \Delta \right\|_{\infty} +\frac{Ms+1}{n}  \left\| \sum_{i=1}^n {\Psi}_{S^0}(t_i) D \Delta \right\|_{\infty} \\
\leq & \frac{2Ms+2}{n}\|nBD\Delta\|_{\infty} = 2(Ms+1)BD\Delta,
\end{aligned}
\end{equation}
where the first inequality follows from \eqref{eqn::GC}, the last inequality follows from the bounds in Assumption~\ref{asmp::psibound}, and the second to last inequality follows from the following inequality: for $k \in S^0$ and $m=1,\ldots, M$,
\begin{equation}\label{eqn::meanvalue}
\small
\begin{aligned}
|\hat{\Psi}_{km}(t_i)-\Psi_{km}(t_i)| & = \left| \int_0^{t_i} \psi_m(\hat{X}_k(u)) \, du - \int_0^{t_i} \psi_m({X}^*_k(u)) \, du \right|\\
& =\left| \int_0^{t_i} \{\psi_m(\hat{X}_k(u)) - \psi_m({X}^*_k(u))\} \, du \right|\\
& \leq \left| \int_0^{t_i} |D \{\hat{X}_k(u)-{X}^*_k(u)\}| \, du \right| \\
& \leq \left\{\int_0^{t_i} D^2 \, du \right\}^{1/2} \left\{ \int_0^{t_i} (\hat{X}_k(u)-{X}^*_k(u))^2\,du \right\}^{1/2} \\
& \leq D \vertiii{\hat{X}_k-X^*_k} \leq D\Delta.
\end{aligned}
\end{equation}
Here the first inequality follows from the mean-value theorem and the bounds in Assumption~\ref{asmp::psibound}.

Now, from \eqref{eqn::GC},
\begin{equation}\label{eqn::AnA}
\| {A}_n- A\|_2 \leq  (Ms+1) \| A_n - A\|_{\infty}  \leq (Ms+1)\frac{BD+B^2}{6n^2},
\end{equation}
where for each element of the matrix $A_n - A = \frac{1}{n}\sum_{i=1}^n \Psi_{S^0}(t_i) \Psi_{S^0}^{\T}(t_i) - \int_0^1 \Psi_{S^0}(t)\Psi^{\T}_{S^0}(t) \, dt$,
\begin{equation*}
\begin{aligned}
& \left| \frac{1}{n}\sum_{i=1}^n \Psi_{km_1}(t_i)\Psi_{lm_2}(t_i) - \int_0^1 \Psi_{km_1}(t)\Psi_{lm_2}(t) \, dt \right| \\
 \leq  & \frac{ \left|\left\{\Psi_{km_1}(u)\Psi_{lm_2}(u)\right\}^{''}\right|  } {12n^2} \leq \frac{ \left|2\Psi'_{km_1}(u)\Psi'_{lm_2}(u)+ \Psi^{''}_{km_1}(u)\Psi_{lm_2}(u)+\Psi'_{km_1}(u)\Psi^{''}_{lm_2}(u) \right| } {12n^2} \\
\leq & \frac{ 2B^2+BD+BD  } {12n^2}  = \frac{BD+B^2}{6n^2},
\end{aligned}
\end{equation*}
where derivatives are taken with respect to $t$. 
By the trapezoid rule on a uniform grid, the first inequality holds for some $u \in [0,1]$.
The second inequality makes use of the bounds in Assumption~\ref{asmp::psibound}, which imply that
$$|\Psi_{km}'(t)|=\left| \left(\int_0^t \psi_{km}(s)\, ds\right)'\right|=| \psi_{km}(t)| \leq B \quad $$
and
$$ \quad |\Psi_{km}''(t)|=\left| \left(\int_0^t \psi_{km}(s)\, ds\right)^{''}\right|=| \psi'_{km}(t)| \leq D. $$

In summary,  combining \eqref{eqn::A_decomp}, \eqref{eqn::Ahatn}, and \eqref{eqn::AnA},
\begin{equation*}
\begin{aligned}
\Lambda_{\min}(\hat{A}_n) \geq & \Lambda_{\min}(A) - \left( 2BD\Delta+ \frac{BD+B^2}{6 n^2} \right)(Ms+1) \\
\geq & C_{\min} - \left( 2BD\Delta+ \frac{BD+B^2}{6 n^2} \right)(Ms+1).
\end{aligned}
\end{equation*}

The upper bound for $\Lambda_{\max}(\hat{A}_n)$ and the lower bound for $\Lambda_{\min}\left( \frac{1}{n} \sum_{i=1}^n \hat{\Psi}_{k}(t_i) \hat{\Psi}_{k}^{\T}(t_i)  \right)$ can be established in a similar manner.
\end{proof}

\subsection{Proof of Lemma~\ref{lmm::irrepresentability_sample}}\label{sec::irrepresent_sample}

\begin{proof}
Define $A$, $A_n$, and $\hat{A}_n$ as in the proof for Lemma~\ref{lmm::coherence_sample}. We let $F=\int_0^1 \Psi_k \Psi_{S^0}^{\T} \, dt$,  $F_n =\sum_{i=1}^n\Psi_k(t_i) \Psi_{S^0}^{\T}(t_i)/n$, and $\hat{F}_n =\sum_{i=1}^n\hat{\Psi}_k(t_i) \hat{\Psi}_{S^0}^{\T}(t_i)/n$. $F, F_n,$ and $\hat{F}_n$ are $M \times (Ms+1)$ matrices. We let $\hat{C}_{\min}$  denote the lower bound of $\Lambda_{\min}(\hat{A}_n)$ established in Lemma~\ref{lmm::coherence_sample}, i.e.,  
$$ \hat{C}_{\min} \equiv C_{\min} - \left( 2BD\Delta+ \frac{BD+B^2}{6 n^2} \right)(Ms+1). $$
To prove the result, we need to bound $\|\hat{F}_n\hat{A}^{-1}_n \|_2$. Note that 
\begin{equation*}
\begin{aligned}
\|\hat{F}_n\hat{A}^{-1}_n \|_2 & \leq \|  \hat{F}_n\hat{A}^{-1}_n - \hat{F}_n{A}^{-1}_n +\hat{F}_n{A}^{-1}_n - {F}_n{A}^{-1}_n +{F}_n{A}^{-1}_n\|_2\\
& \leq \|  \hat{F}_n( \hat{A}^{-1}_n-{A}^{-1}_n)\|_2 +\| (\hat{F}_n-F_n) {A}^{-1}_n\|_2 + \|{F}_n{A}^{-1}_n\|_2 \\ \equiv & \|\textrm{I}\|_2 + \|\textrm{II}\|_2 + \|\textrm{III}\|_2.
\end{aligned}
\end{equation*}
Using sub-multiplicity of the $\ell_2$-norm of matrices,
$$\| \textrm{I}\|_2^2 \leq  \|\hat{F}_n\|_2^2 \|\hat{A}_n^{-1}-{A}_n^{-1}\|_2^2.$$
\begin{align*}
\intertext{Applying \eqref{eqn::norm} to $\hat{F}_n$, we get}
\| \textrm{I}\|_2^2 \leq & M(Ms+1) \left(\max_{i,j} \hat{F}_{n,ij}^2\right)  \|\hat{A}_n^{-1}-{A}_n^{-1}\|^2_2. \\
\intertext{Recalling that $\hat{F}_n =\sum_{i=1}^n\hat{\Psi}_k(t_i) \hat{\Psi}_{S^0}^{\T}(t_i)/n$ and that $|\hat{\Psi}_{km}(t_i)|\leq B$,}
\| \textrm{I}\|_2^2 \leq & M(Ms+1) (\sum_{i=1}^n  B^2/n)^2 \|\hat{A}_n^{-1}-{A}_n^{-1}\|^2_2. \\
\intertext{Note that $\hat{A}_n^{-1}-{A}_n^{-1}=\hat{A}_n^{-1}({A}_n-\hat{A}_n)A_n^{-1}$. Thus,}
\| \textrm{I}\|_2^2 \leq & M(Ms+1) B^4 \|\hat{A}^{-1}_n\|^2_2 \|\hat{A}_n - A_n\|_2^2 \|A_n^{-1}\|_2^2\\
\leq & M(Ms+1) B^4 \hat{C}_{\min}^{-2} \| \hat{A}_n -A_n\|^2_2 {C}_{\min}^{-2} \\
\leq & M(Ms+1) B^4 \{2(M s+1) D B \Delta\}^2 \hat{C}_{\min}^{-2} {C}_{\min}^{-2},\\
\equiv & c_1 \hat{C}_{\min}^{-2} M(Ms+1)^3 \Delta^2,
\end{align*}
where the last two inequalities follow from the proof of Lemma~\ref{lmm::coherence_sample}.

Next, note that
\begin{align*}
\|\textrm{II}\|_2^2 = & \| (\hat{F}_n- F_n) A_n^{-1} \|_2^2 \\
\leq & C_{\min}^{-2} \left\| \frac{1}{n}\sum_{i=1}^n  \hat{\Psi}_k(t_i) \hat{\Psi}^{\T}_{S^0}(t_i) - \frac{1}{n}\sum_{i=1}^n  {\Psi}_k(t_i) {\Psi}^{\T}_{S^0}(t_i)\right\|_2^2 \\
\leq & C_{\min}^{-2} \left\| \frac{1}{n}\sum_{i=1}^n  \hat{\Psi}_k(t_i) \left\{\hat{\Psi}^{\T}_{S^0}(t_i)- {\Psi}^{\T}_{S^0}(t_i) \right\}\right\|_2^2 +\\
& C_{\min}^{-2} \left\| \frac{1}{n}\sum_{i=1}^n \left\{ \hat{\Psi}_k(t_i)-\Psi_k(t_i) \right\} {\Psi}^{\T}_{S^0}(t_i) \right\|_2^2 \\
\leq & 2 C_{\min}^{-2} B^2 D^2 \Delta^2 M(Ms+1)\\
\equiv & c_2 M(Ms+1)\Delta^2,
\end{align*}
where the first inequality follows from sub-multiplicity of norms of matrices, and the last from \eqref{eqn::norm}, \eqref{eqn::meanvalue}, and the bounds in Assumption~\ref{asmp::psibound}.

Finally,
\begin{equation*}
\begin{aligned}
\|\textrm{III}\|_2 = & \| F_n A^{-1}_n\|_2  = \| F_n(A^{-1}_n-A^{-1}) + (F_n-F) A^{-1} +FA^{-1}\|_2\\
\leq & \xi + \|F_n\|_2 \| A^{-1}_n - A^{-1}\|_2 +\|(F_n-F) A^{-1}\|_2\\
\leq  & \xi + \left\{ M(Ms+1)B^4  \| A^{-1}_n - A^{-1}\|_2^2 \right\}^{1/2} + \{\|F_n-F\|_2^2 C_{\min}^{-2}\}^{1/2}\\
\leq  & \xi + \left\{ M(Ms+1)B^4  \|A^{-1}_n\|_2^2 \| A_n - A\|_2^2 \|A^{-1}\|_2^2 \right\}^{1/2}+\{\|F_n-F\|_2^2 C_{\min}^{-2}\}^{1/2}\\
\leq  & \xi + \left\{ M(Ms+1)B^4 \hat{C}_{\min}^{-2} {C}_{\min}^{-2} \| A_n - A\|_2^2 \right\}^{1/2}+\{\|F_n-F\|_2^2 C_{\min}^{-2}\}^{1/2}\\
\leq  & \xi + \left\{ M(Ms+1)B^4 \hat{C}_{\min}^{-2} {C}_{\min}^{-2} (Ms+1)^2 \frac{BD+B^2}{6n^2} \right\}^{1/2}+\\
& \left\{ M(Ms+1){C}_{\min}^{-2}\frac{BD+B^2}{6n^2} \right\}^{1/2}\\
\leq & \xi + \left\{ c_3 M(Ms+1)^3/{6n^2} \right\}^{1/2},
\end{aligned}
\end{equation*}
where the first inequality follows from Assumption~\ref{asmp::irrepresentability_pop} in the main paper and  the second to last inequality follows from \eqref{eqn::AnA}.

In summary,
\begin{equation*}
\begin{aligned}
& \left\| \frac{1}{n} \sum_{i=1}^n\hat{\Psi}_k(t_i) \hat{\Psi}_{S^0}^{\T}(t_i)  \left(\frac{1}{n} \sum_{i=1}^n\hat{\Psi}_{S^0}(t_i)\hat{\Psi}_{S^0}^{\T}(t_i)  \right)^{-1} \right\|_2  \leq \\
 & \xi + \left\{c_1 M(Ms+1)^3 \Delta^2\right\}^{1/2}+\left\{c_2 M(Ms+1)\Delta^2\right\}^{1/2}+ \left\{ c_3 M(Ms+1)^3/{6n^2} \right\}^{1/2}. 
\end{aligned}
\end{equation*}
where $c_1, c_2, c_3$ are constants. 
\end{proof}

\subsection{Proof of Lemma~\ref{lmm::deviation}}\label{sec::deviation}

\begin{proof}
For $k=1,\ldots, p$,
\begin{equation*}
\begin{aligned}
&\left\| \frac{1}{n} \sum_{i=1}^n \hat{\Psi}_k(t_i) Y_{ij} - \frac{1}{n} \sum_{i=1}^n \hat{\Psi}_k(t_i) \hat{\Psi}_{S^0}^{\T}(t_i) \theta^*_{S^0} \right\|_2 \\
= & \frac{1}{n} \left\| \sum_{i=1}^n \hat{\Psi}_k(t_i) X^*_{j}(t_i) + \sum_{i=1}^n \hat{\Psi}_k(t_i)\epsilon_{ji} - \sum_{i=1}^n \hat{\Psi}_k(t_i) {\Psi}_{S^0}^{\T}(t_i) \theta^*_{S^0} +  \right.\\
& \left. \sum_{i=1}^n \hat{\Psi}_k(t_i) {\Psi}_{S^0}^{\T}(t_i) \theta^*_{S^0} - \sum_{i=1}^n \hat{\Psi}_k(t_i) \hat{\Psi}_{S^0}^{\T}(t_i) \theta^*_{S^0} \right\|_2 \\
\leq & \left\|\frac{1}{n} \sum_{i=1}^n \hat{\Psi}_k(t_i) \{ X_{j}^*(t_i)-\Psi_{S^0}^{\T}(t_i)\theta^*_{S^0}\} \right\|_2 + \left\|\frac{1}{n} \sum_{i=1}^n \hat{\Psi}_k(t_i) \{{\Psi}_{S^0}^{\T}(t_i)-\hat{\Psi}_{S^0}^{\T}(t_i)\} \theta^*_{S^0} \right\|_2 + \\
&  \left\|\frac{1}{n} \sum_{i=1}^n \hat{\Psi}_k(t_i)\epsilon_{ji} \right\|_2\\
\equiv & \|\textrm{I}\|_2+\|\textrm{II}\|_2+\|\textrm{III}\|_2.
\end{aligned}
\end{equation*}

First, applying the Cauchy-Schwarz inequality to $\|\textrm{I}\|_2^2$,
\begin{align*}
\| \textrm{I}\|_2^2 \leq & \sum_{m=1}^M \left[ \frac{1}{n^2} \sum_{i=1}^n \hat{\Psi}^2_{km}(t_i) \sum_{i=1}^n \left\{X_j^*(t_i)- \Psi_{S^0}^{\T}(t_i) \theta^*_{S^0} \right\}^2 \right].\\
\intertext{From the bounds in Assumption~\ref{asmp::psibound} and \eqref{eqn::truncation_error_dis} ,}
\| \textrm{I}\|_2^2 \leq & M\left\{ \frac{1}{n^2} \left(n B^2\right) \left( s^2 n Q M^{-2\beta_2}\right) \right\} = s^2 M^{-2\beta_2 +1} Q B^2.  
\end{align*}

Next, note that $\hat{\Psi}_0(t_i) -  \Psi_0(t_i) = t_i- t_i=0$, we have $\left\{ \Psi_{S^0}^{\T}(t_i) - \hat{\Psi}_{S^0}^{\T}(t_i) \right\}\theta_{S^0}^*=\left\{ \Psi_{S}(t_i) - \hat{\Psi}_{S}(t_i) \right\}^{\T}\theta_{S}^*$.  Thus applying the Cauchy-Schwarz inequality to $\|\textrm{II}\|_2^2$, 
\begin{align*}
\|\textrm{II}\|_2^2 \leq &  \sum_{m=1}^M \left(  \frac{1}{n^2}\sum_{i=1}^n \hat{\Psi}^2_{km}(t_i) \sum_{i=1}^n\left[ \left\{ \Psi_{S}(t_i) - \hat{\Psi}_{S}(t_i) \right\}^{\T}\theta_{S}^*\right]^2 \right).\\
\intertext{Applying the norm inequality $a^{\T} b \leq \|a\|_{\infty} \|b\|_1$ to $\left\{ \Psi_{S}(t_i) - \hat{\Psi}_{S}(t_i) \right\}^{\T}\theta_{S}^*$ and using the inequality \eqref{eqn::meanvalue} as well as the bounds in Assumption~\ref{asmp::psibound}, we get}
\|\textrm{II}\|_2^2 \leq & M\left\{\frac{1}{n^2} n B^2 \sum_{i=1}^n \|\theta_{S}^*\|_1^2 D^2 \Delta^2  \right\} \leq  M B^2 D^2  \|\theta^*_{S}\|_1^2 \Delta^2. 
\end{align*}

Finally,  $\textrm{III} = \frac{1}{n} \sum_{i=1}^n \hat{\Psi}_k(t_i) \epsilon_{ji}$ is an $M$-vector. For each $m=1,\ldots, M$, we let $g(\epsilon_j/\sigma) = \sum_{i=1}^n \hat{\Psi}_{km}(t_i) \epsilon_{ji}/n$. Then, for $a,b \in \mathbb{R}^p$,
\begin{equation*}
\begin{aligned}
|g(a)-g(b)| = &  \left|\sigma \sum_{i=1}^n \hat{\Psi}_{km}(t_i)(a_i-b_i)/n\right| \\
\leq  & \frac{\sigma}{n} \left\{ \sum_{i=1}^n \hat{\Psi}^2_{km}(t_i) \right\}^{0.5} \|a-b\|_2 \leq \frac{\sigma }{n}\sqrt{nB^2} \|a-b\|_2.
\end{aligned}
\end{equation*}
This shows that  $g(\cdot)$ is an $L_3$-Lipshitz function with  $L_3=\sigma B/\sqrt{n}$.  Note that $\mathbb{E} g(\epsilon_j/\sigma)=0$. Thus, by Theorem 5.6 in \cite{boucheron2013} presented in Section~\ref{sec::proof_local}, we have 
\begin{equation*}
\text{Pr}(|g(\epsilon_j/\sigma)|\geq v) \leq 2\exp\{-v^2n/(2B^2\sigma^2) \}.
\end{equation*}
Letting $v=n^{\alpha/2-0.5}$, $\|\textrm{III}\|_2^2 \leq n^{\alpha-1} M $ holds with probability at least $1-2M\exp\{-n^{\alpha}/(2B^2\sigma^2) \}$.

Combining all of the pieces, we find that 
\begin{equation*}
\|\textrm{I}\|_2+\|\textrm{II}\|_2+\|\textrm{III}\|_3 \leq \eta \equiv M^{1/2} \left\{s M^{-\beta_1} Q^{1/2} B+BD\|\theta^*_{S}\|_1 \Delta +n^{\frac{\alpha}{2}-\frac{1}{2} }  \right\}
\end{equation*}
with probability at least $1-2M\exp\{-n^{\alpha}/(2B^2\sigma^2) \}$. 

For $k=0$, 
\begin{equation*} \small
\begin{aligned}
& \left\| \frac{1}{n} \sum_{i=1}^n \hat{\Psi}_0(t_i) \left\{ Y_{ij} - \hat{\Psi}_{S^0}^{\T}(t_i) \theta^*_{S^0} \right\}\right\|_2 =  \left\| \frac{1}{n} \sum_{i=1}^n t_i \left\{ Y_{ij} - \hat{\Psi}_{S^0}^{\T}(t_i) \theta^*_{S^0} \right\}\right\|_2 \\
\leq & \left\|\frac{1}{n} \sum_{i=1}^n t_i \{ X_{j}^*(t_i)-\Psi_{S^0}^{\T}(t_i)\theta^*_{S^0}\} \right\|_2 + \left\|\frac{1}{n} \sum_{i=1}^n t_i \{{\Psi}_{S^0}^{\T}(t_i)-\hat{\Psi}_{S^0}^{\T}(t_i)\} \theta^*_{S^0} \right\|_2 +\\
& \left\|\frac{1}{n} \sum_{i=1}^n t_i\epsilon_{ji} \right\|_2.
\end{aligned}
\end{equation*}
Recall that $t \in [0,1]$ and, without loss of generality, let $B\geq 1$. Thus, we can see from the same argument that $\left\| \frac{1}{n} \sum_{i=1}^n t_i \left\{ Y_{ij} - \hat{\Psi}_{S^0}^{\T}(t_i) \theta^*_{S^0} \right\}\right\|_2\leq \eta$ holds with the same probability.

\end{proof}

\section{DETAILS ABOUT DATA GENERATION}\label{sec::data_appendix}
In this section, we provide details about the parameters used for generating data in Section~\ref{sec::sparse} of the main paper (see Equation~\ref{eqn::FHN}). 
Three pairs of variables, $(X_1,X_2), (X_3,X_4), (X_5,X_6)$, are solutions of \eqref{eqn::FHN} in the main paper with the following parameters and initial values: 
\begin{enumerate}
\item $(X_1,X_2)$ are generated according to \eqref{eqn::FHN} from  the main paper with $\theta_{1,0}=0$, $\theta_{1,1}=(1.2,0.3,-0.6)^{\T}$, $\theta_{1,2}=(0.1,0.2,0.2)^{\T}$, $\theta_{2,0}=0.4$, $\theta_{2,1}=(-2,0,0.4)^{\T}$, $\theta_{2,2}=(0.5,0.2,-0.3)^{\T}$, and initial values $X_1(0)=-2, X_2(0)=2$. 
\item $(X_3,X_4)$ are generated according to \eqref{eqn::FHN} from  the main paper with $\theta_{3,0}=-0.2, \theta_{3,3}=(0,0,0)^{\T}, \theta_{3,4}=(-0.3,0.4,0.1)^{\T}, \theta_{4,0}= -0.2, \theta_{4,3}=(0.2,-0.1,-0.2)^{\T}, \theta_{4,4}=(0,0,0)^{\T}$, and initial values $X_3(0)=2, X_4(0)=-2$. 
\item $(X_5,X_6)$ are generated according to \eqref{eqn::FHN} from  the main paper with $\theta_{5,0}=0.05$, $\theta_{5,5} = (0,0,0)^{\T}$, $\theta_{5,6}=(0.1,0,-0.8)^{\T}$, $\theta_{6,0}=-0.05$, $\theta_{6,5}=(0,0,0.5)^{\T}$, $\theta_{6,6}=(0,0,0)^{\T}$, and initial values $X_5(0)=-1.5, X_6(0)=1.5$. 
\end{enumerate}
Solution trajectories of $X_1, \ldots, X_6$ are shown in Figure~\ref{fig::solution}.  For $X_7, \ldots, X_{10}$, we drew the initial values $X_j(0), j=7, \ldots, 10,$ and the $\theta_{j,0}, j=7,\ldots, 10$, from a normal distribution. All other parameters were set to zero, so that $X_7, \ldots, X_{10}$ represent ``noise" variables. 
The directed graph of $X_1,\ldots, X_{10}$ is showing in Figure~\ref{fig::network}. 

\begin{figure}
\begin{center}
\includegraphics[scale=0.35]{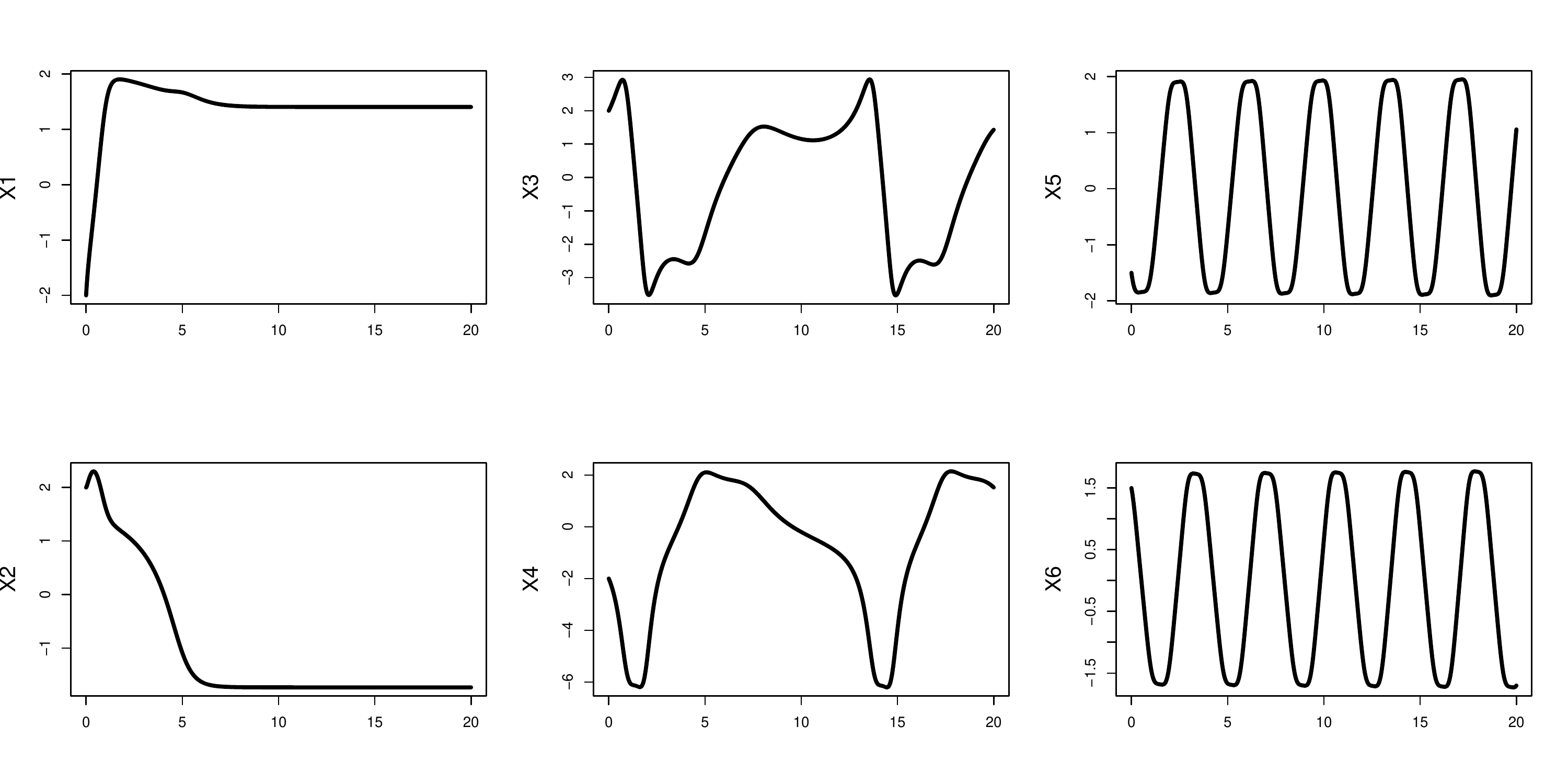}
\caption[]{The curves $X_1,\ldots, X_{6}$ on $[0,20]$ described in Section~\ref{sec::sparse} of the main paper and Section~\ref{sec::data_appendix} of the supplementary material.} 
\label{fig::solution}
\end{center}
\end{figure}
\end{appendices}

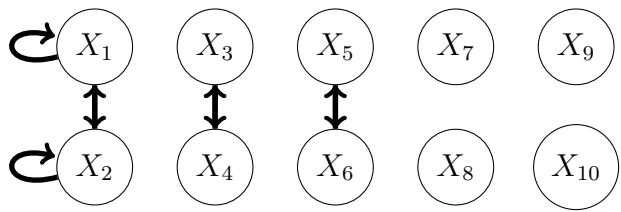
\begin{figure}[ht]
\begin{center}
\begin{tikzpicture}
  [scale=.8,auto=left,  every node/.style={circle,fill=blue!20}]

  \node[draw=black, fill=white,circle]  (X1) at (1,3) {$X_1$};
  \node[draw=black, fill=white,circle]  (X2) at (1,1)  {$X_2$};
  \node[draw=black,fill=white,circle] (X3) at (3,3) {$X_3$};
  \node[draw=black, fill=white,circle]  (X4) at (3,1)  {$X_4$};
  \node[draw=black, fill=white,circle]  (X5) at (5,3)  {$X_5$};
   \node[draw=black, fill=white,circle]  (X6) at (5,1)  {$X_6$};
\node[draw=black, fill=white,circle]  (X7) at (7,3)  {$X_7$};
   \node[draw=black, fill=white,circle]  (X8) at (7,1)  {$X_8$};
  \node[draw=black, fill=white,circle]  (X9) at (9,3)  {$X_9$};
   \node[draw=black, fill=white,circle]  (X10) at (9,1)  {$X_{10}$};

  \foreach \from/\to in {X1/X2,X2/X1, X3/X4,X4/X3,X5/X6,X6/X5}
	 \draw[->, line width=2] (\from) --  (\to);
\path
    (X1) edge [loop left, line width=2]  (X1)
     (X2) edge [loop left, line width=2] (X2);
\end{tikzpicture}
\end{center}
\caption{The network of $\{X_1, \ldots, X_{10}\}$. A directed edge $j \to k$ indicates that the $j$th node regulates the $k$th node.}
\label{fig::network}
\end{figure}

\end{document}